\documentclass[a4paper,11pt]{article}
\usepackage[utf8]{inputenc}
\usepackage{amsthm,amsmath,amssymb,amsfonts,tikz}
\usepackage{color}

\usepackage[pdftex, bookmarksnumbered=true]{hyperref}

\usepackage{titlesec}
\titleformat*{\section}{\Large \normalfont \bfseries}
\titleformat*{\subsection}{\large \normalfont \bfseries}
\usepackage{setspace}

\newcommand{\nc}{\newcommand}

\nc{\be}{\begin{equation}} \nc{\ee}{\end{equation}}
\nc{\ba}{\begin{array}} \nc{\ea}{\end{array}}
\nc{\bea}{\begin{eqnarray}} \nc{\eea}{\end{eqnarray}}
\nc{\ny}{\nonumber}
\nc{\ra}{\rangle}
\nc{\la}{\langle}
\nc{\lb}{\left(}
\nc{\rb}{\right)}
\nc{\pt}{\partial}
\nc{\D}{\Delta}
\nc{\kr}{\mathcal}
\nc{\lmb}{\lambda}
\nc{\sg}{\sigma}
\nc{\bbl[1]}{\lbrack #1\rbrack{}}
\nc{\ep}{\epsilon}
\nc{\ora}{\overrightarrow}
\nc{\td}{\widetilde}
\nc{\tq}{\theta}
\nc{\seteq}{\mathbin{:=}}
\nc{\simto}{\xrightarrow{\,\sim\,}}
\nc{\Z}{{\mathbb Z}}
\nc{\NSR}{\textsf{NSR}}
\nc{\sfF}{\textsf{F}}
\nc{\Vir}{\textsf{Vir}}
\nc{\Asl}{\widehat{\mathfrak{sl}}}
\nc{\rmF} {\mathrm{F}}
\nc{\calM}{\mathcal{M}}
\nc{\calA}{\mathcal{A}}
\nc{\calF}{\mathcal{F}}
\nc{\calL}{\mathcal{L}}
\nc{\calU}{\mathcal{U}}
\nc{\calE}{\mathcal{E}}
\nc{\calV}{\mathcal{V}}
\nc{\calH}{\mathcal{H}}
\nc{\sfV}{\textsf{V}}
\nc{\tr}{\mathrm{Tr}}
\nc{\red}{\color{red}}
\nc{\NS}{\scriptscriptstyle{\textsf{NS}}}
\nc{\R}{\scriptscriptstyle{\textsf{R}}}
\nc{\nsr}{\scriptscriptstyle{\textsf{NSR}}}
\nc{\vac}{\varnothing}
\nc{\pure}{\mathrm{pure}}
\nc{\mf}{\mathfrak}
\nc{\q}{\color{red}}
\nc{\G}{\mathsf{G}}
\nc{\BNS}{\textsf{NS}}
\nc{\BR}{\textsf{R}}
\nc{\Q}{\mathbf{Q}}

\numberwithin{equation}{section}
\newtheorem{thm}{Theorem}[section]
\newtheorem{prop}{Proposition}[section]
\newtheorem{lemma}{Lemma}[section]
\newtheorem{conj}{Conjecture}[section]

\newtheorem{Remark}{Remark}[section]

\title{B\"acklund transformation of Painlev\'e III($D_8$) $\tau$ function}
\date{}
\author{M.~A.~Bershtein, A.~I.~Shchechkin}
\begin{document}

%\begin{titlepage}
%{\Large{\center
%  {\bf National Research University\\[0.5cm]
%  Higher School of Economics\\[2cm]
%  
%  Department of mathematics\\[2.5cm]
%  
%  Coursework}\\[1.5cm]
% 
%  Theme: Bilinear relations on q-deformed conformal blocks
%  }\\[4cm]
%
%  {\flushleft
%  \leftskip=10cm
%   Author\\
%  Anton Shchechkin\\[1cm]
%  Supervisor \\
%  Dr. Mikhail Bershtein\\[5cm]}

%{\center Moscow, 2015\\}}
%\end{titlepage}
%\newpage
\maketitle
\begin{abstract}\vspace*{2pt}
We study explicit formula (suggested by Gamayun, Iorgov, Lisovyy) for Painlev\'e III($D_8$) $\tau$~function in terms of Virasoro conformal blocks with central charge $1$. The Painlev\'e equation has two types of bilinear forms, we call them Toda-like and Okamoto-like. We obtain these equations from the representation theory using an embedding of a direct sum of two Virasoro algebra in a certain superalgebra. These two types of bilinear forms correspond to Neveu-Schwarz sector and Ramond sector of this algebra. We also obtain $\tau$ functions of algebraic solutions of Painlev\'e III($D_8$) from the special representations of the Virasoro algebra of highest weight $(n+1/4)^2$.
\end{abstract}

\tableofcontents

\newpage

\section{Introduction}
This paper is a sequel to \cite{KGP}. We continue our study of the relation between Painlev\'e equations and conformal field theory. 
In this paper we restrict ourselves to the most degenerate case of Painlev\'e III equation.
This equation has different names: it is called Painlev\'e III($D_8$) equation in the geometric approach (see e.g. \cite{OKSO}),
it is also called Painlev\'e III$_3$ equation (see  e.g. \cite{GIL1302}) and it is also equivalent to radial sine-Gordon equation (see e.g. \cite{Fokas}).

Gamayun, Iorgov, Lisovyy in the paper \cite{GIL1302} (following their previous work \cite{GIL1207}) suggested that $\tau$ function of this equation has the form 
\begin{equation}
\tau(\sigma,s|z)=\sum_{n\in\mathbb{Z}}C(\sigma+n)s^n  \mathcal{F}((\sigma+n)^2|z), \label{eq:GIL}
\end{equation}
where $s,\sigma$ are integration constants, the function $\mathcal{F}(\Delta|z)$ denotes Whittaker limit of Virasoro conformal block in module
with highest weight $\Delta$ and central charge $c=1$ and the function $C(\sigma)=1/(\G(1-2\sigma)\G(1+2\sigma))$,
where $\mathsf{G}$ is a Barnes $\mathsf{G}$-function. This formula was proven in \cite{ILT} and \cite{KGP} by different methods.

On the other hand it is known that this Painlev\'e equation has a B\"acklund transformation $\pi$ of order two.
The main topic of this paper is a relation between decomposition \eqref{eq:GIL} and this B\"acklund transformation.
Our initial motivation was a $q$-deformation of formula \eqref{eq:GIL}, corresponding results are reported in the separate paper \cite{BS:2016:1}.

Let us discuss the content of the paper. In Section \ref{sec:Painleve} we recall  necessary definitions and notations of Painlev\'e III($D_8$) equation,
its $\tau$ function and B\"acklund transformation. Relation between $\tau$ and transformed $\tau_1=\pi(\tau)$ is written in form of bilinear relations.
We have two types of bilinear relations, namely  Okamoto-like and Toda-like 
\begin{equation}
\left\{\begin{aligned} &D_{[\log z]}^2(\tau,\tau_1)-\frac12\left(z\frac{d}{dz}-\frac18\right)(\tau\tau_1)=0,
 \\
 &D_{[\log z]}^3(\tau,\tau_1)-\frac12\left(z\frac{d}{dz}-\frac18\right)D^1_{[\log z]}(\tau,\tau_1)=0,
 \end{aligned}\right.\qquad\qquad
\left\{\begin{aligned}    
D^2_{[\log z]}(\tau,\tau)=2z^{1/2}\tau_1^2, \\
D^2_{[\log z]}(\tau_1,\tau_1)=2z^{1/2}\tau^2, 
\end{aligned}\right. \label{eq:bilin}
\end{equation}
where $D_{[\log z]}^k$ are Hirota differential operators \eqref{Hiru} (see Propositions \ref{prop:Okamoto} and \ref{prop:Toda}). We obtain these equations by a standard in the Painlev\'e theory method,
the only difference is the fact that $\pi$ has order 2 instead of infinite order B\"acklund transformations used usually. We also prove converse statement i.e. to what extent equations \eqref{eq:bilin} determine  Painlev\'e III($D_8$) $\tau$ function.

In Section \ref{sec:Kiev} we recall necessary notations from the representation theory of the Virasoro algebra in order to state formula \eqref{eq:GIL}.
Then we show that 
\begin{equation}
\tau_1(\sg,s|z)\propto\tau(\sg-1/2,s|z),
\label{eq:BackTau}
\end{equation}
where $\propto$ stands for constant (with respect to $z$) proportionality. In the Section \ref{sec:FxNSR} we discuss interpretation of bilinear equations \eqref{eq:bilin} (on $\tau $ function \eqref{eq:GIL}) in the framework of  representation theory of the Virasoro algebra. As in the paper \cite{KGP} the main tool is the embedding $\Vir\oplus \Vir \subset \sfF\oplus\NSR$ of a direct sum of two Virasoro algebras
into a sum of Majorana fermion and Super Virasoro algebra. 
Actually using Neveu--Schwarz sector of the algebra $\sfF\oplus\NSR$ we already proved (up to some details discussed in Subsection \ref{ssec:Todaobt})
in \cite{KGP}  that right side of \eqref{eq:GIL} satisfies Toda-like equations. This proves formula \eqref{eq:GIL}, actually this is slight simplification of the proof in \cite{KGP} where we used
another bilinear equation of order 4. In Subsection \ref{ssec:Okamotoobt} we show that Ramond sector of the algebra  $\sfF\oplus\NSR$ gives
Okamoto-like equations for right side of \eqref{eq:GIL}. 
%, this is one of the results of the paper.

Painlev\'e III($D_8$) equation has two algebraic solutions. Corresponding $\tau$ functions have the form $\tau(z)\propto z^{1/16}e^{\mp 4\sqrt{z}}$. 
In Subsection \ref{ssec:algsol} we give interpretation of these $\tau$ functions in the framework of representation theory of Virasoro algebra.
Namely these $\tau$ functions correspond to special representations of highest weights $(n+1/4)^2,\,n\in\mathbb{Z}$, studied by Al. Zamolodchikov in \cite{Zam:1987}.
Note also that in this case $\tau$ function coincides with special case of dual partition function introduced by Nekrasov and Okounkov in \cite{NO:2003}.

In this paper all continuous variables are considered to belong to the field $\mathbb{C}$ unless otherwise stated. All representations and algebras are considered over the field $\mathbb{C}$.

\section{Painlev\'e III($D_8$) equation}
\label{sec:Painleve}
\subsection{Hamiltonian and $\tau$ form of Painlev\'e III($D_8$) equation}
\label{ssec:Ham}

We recall several facts about one of the simplest Painlev\'e equation --- Painlev\'e III($D_8$) (or Painlev\'e III$_3$) 
following \cite{GIL1302}, \cite{OKSO}. 

The Painlev\'e III($D_8$) equation on function $w(z)$ has the form
\begin{equation}\label{piii}
\frac{d^2w}{dz^2}=\frac{1}{w}\left(\frac{dw}{dz}\right)^2 -
 \frac{1}{z}\frac{dw}{dz}\,+\frac{2w^2}{z^2}-\frac{2}{z}.
 \end{equation}
 Note that in work \cite{OKSO} rescaled $w$ and $z$ are used.

 We now proceed to the Hamiltonian (or $\zeta$) form of Painlev\'e III($D_8$). 
 The Painlev\'e equations can be rewritten as non-autonomous Hamiltonian systems.
 It means that they can be obtained by eliminating an auxiliary momentum $p(z)$ from the equations
\[
\frac{dw}{dz}=\frac{\partial H}{\partial p},\qquad \frac{dp}{dz}=-\frac{\partial H}{\partial w}
\]
where Hamiltonian $H(z)$ for the Painlev\'e III($D_8$) equation has the form
\begin{align}
\zeta=zH=&p^2w^2-w-z/w  \label{ham}.
\end{align}
It is also convenient to use the function $\zeta(z)=zH(z)$, which is just Hamiltonian with respect to the time $\log z$. 
We will below denote by dot differentiation by $z$ and by prime differentiation by $\log z$.
Hamilton equations in terms of $p$, $w$ read
\begin{equation}
\dot{w}=2pw^2/z, \qquad \dot{p}=-2p^2w/z+1/z-1/w^2 \label{Hameq}.  
\end{equation}

Remark that if we know function $\zeta(z)$ on trajectories of motion then we can find $w(z)$ and $p(z)$. Differentiating \eqref{ham}
once and twice and using Hamilton equations to differentiate $p(z)$ and $w(z)$ we could express these functions by formulas
\begin{equation}
w(z)=-\frac{1}{\dot\zeta(z)}, \quad p(z)=\frac{z\ddot\zeta(z)}2. \label{conn}
\end{equation}
Substituting these expressions into \eqref{ham} we get Hamiltonian (or $\zeta$) form of Painlev\'e III($D_8$) equation
\begin{equation}
(z\ddot\zeta(z))^2=4\dot\zeta(z)	^2(\zeta(z)-z\dot\zeta(z))-4\dot\zeta(z) \label{zeta3}.
\end{equation}
In this paper we will consider solutions of this equation except constant $\zeta(z)$ and $\zeta(z)=z+1$ (these are only solution such that  $\ddot{\zeta}= 0$).

Then it could be checked directly that each solution of \eqref{zeta3} corresponds by first formula of \eqref{conn} to solution of \eqref{piii}.
Inversely solution $w(z)$ of \eqref{piii} give us $p(z)$ by first formula of \eqref{Hameq} and then $\zeta(z)$ given by \eqref{ham} satisfy \eqref{zeta3}.
So we have one-to-one correspondence between solutions $w(z)$ of \eqref{piii} and $\zeta(z)$ of \eqref{zeta3}.

\begin{Remark}\label{rem:sinh}
Painlev\'e III($D_8$) equation appears in physical framework for instance as radial sine-Gordon equation on function $v(r)$
(see e.g. \cite[Chapter 3]{Fokas})
\begin{equation}
v_{rr}+\frac{v_r}{r}=1/2\sin 2v\label{eq:rshG}
\end{equation}
One can proceed to this equation from \eqref{piii} by substituting
\begin{equation*}
w(z)/\sqrt{z}= e^{2iv}, \quad z=r^4/4096 
\end{equation*}
\end{Remark}
Let us introduce $\tau$ function by formula
\begin{equation}
\zeta(z)=z\frac{d \log\tau(z)}{dz} \quad \textrm{and inverse} \quad \tau=\exp\left(\int\zeta(z) d\log z\right) .\label{tauzeta}
\end{equation}
Note that $\tau$ function is defined up to a multiplication by a constant factor.

One can obtain equation on $\tau$ function from \eqref{zeta3}.
Differentiate \eqref{zeta3} by $z$ and divide the result by $\ddot\zeta(z)$. Substituting first formula of \eqref{tauzeta} and multiplying by $\tau^2$ 
we obtain bilinear equation on $\tau$ function. It is convenient to write this equation
by use of Hirota differential operators $D^k_{[x]}$. In our paper we use only Hirota derivatives with respect to the logarithm of a variable.
These operators on the functions $f(z),g(z)$ are defined by the formula
\begin{equation}
f(e^{\alpha}z)g(e^{-\alpha}z)=\sum\limits_{k=0}^{\infty}D^{k}_{[\log z]}(f(z),g(z))\frac{\alpha^k}{k!}.\label{Hiru}
\end{equation}
The first examples of Hirota operators are
\begin{equation*}
D^0_{[\log z]}(f(z),g(z))=f(z)g(z),\qquad D^1_{[\log z]}(f(z),g(z))=z\dot{f}(z)g(z)-f(z)z\dot{g}(z). 
\end{equation*}
Then, the $\tau$ form of the Painlev\'e III($D_8$) equation can be written as
\begin{equation}
D^{III}(\tau(z),\tau(z))=0,\quad \text{where} \quad D^{III}=
\frac12 D^{4}_{[\log z]}-z\frac{d}{dz}D^{2}_{[\log z]}+\frac12D^{2}_{[\log z]}+2zD^{0}_{[\log z]} \label{tau3}.
\end{equation}
Because we differentiate \eqref{zeta3} to obtain \eqref{tau3} we have extra solutions of \eqref{tau3}.
More precisely, \eqref{tau3} is equivalent to so-called Painlev\'e III($D_7$) equation
\begin{equation*}
(z\ddot\zeta(z))^2=4(\dot\zeta(z))^2(\zeta(z)-z\dot\zeta(z))-4\dot\zeta(z)+1/\theta_{*} 
\end{equation*}
In this work we will consider only that solutions of \eqref{tau3} which correspond to case $\theta_{*}=\infty$ i.e. Painlev\'e III($D_8$)
in form \eqref{zeta3}. 
These solutions could be distinguished by the asymptotic behavior of $\tau$ function.

The following proposition follows from results proven in \cite{Niles}, (see also book \cite{Fokas} and original papers \cite{McCoy:1977} \cite{Jimboas}, \cite{Nov1985} \footnote{we are grateful to A. Its for the explanation of this point and help with the references}.)
\begin{prop}\label{Jimboas}
There exists two-parametric family of solutions of the equation \eqref{zeta3} such that the asymptotic behavior of corresponding $w(z)$ 
and $\tau(z)$ for $z \rightarrow 0$ are given by 
\begin{align}
w(\sg,\td{s}|z)&=4\sg^2\td{s}z^{2\sg}(1+o(1)),&\label{was}\\
\tau(\sg,\td{s}|z)&\propto z^{\sg^2}\left(1+\frac{z}{2\sg^2}-\frac{\td{s}^{-1}}{(1-2\sg)^2(2\sg)^2}z^{1-2\sg}+o(|z|)\right),&  \label{tauas}
\end{align}
where $\propto$ means constant (with respect to $z$) proportionality and $\sg,\td{s}$ are integration constants which belong to the domain
$0<\operatorname{Re} \sg<1/2,\, \td{s}\neq 0$. Moreover, any solution of \eqref{zeta3} with such asymptotics
belongs to this family and for given $\sg$, $\td{s}$ it is unique. 
\end{prop}

It was also proven in \cite{Niles} that solutions which do not belong to this family can be paramerized by less number of parameters (3 real numbers). Therefore, one can think about the family from Proposition \ref{Jimboas} is a family of generic solutions. 

%This fact is stated and proved in \cite{Niles}.
%We have not found reference where such fact proposed. Presumably it is analogous to the Theorems 1.1 and 3.1 in \cite{Jimboas}.

%One of the aims of this work is investigation instead of \eqref{tau3} equivalent but more simple form of equations on $\tau$ functions.
%This will be done in next subsection.

\subsection{B\"acklund transformation, Okamoto-like and Toda-like equations}
\label{ssec:OT}
The group of B\"acklund transformations of Painlev\'e III($D_8$) equation is $\mathbb{Z}_2$ (see \cite[Sec. 2.3]{OKSO}).
This group is generated by transformation $\pi$ which acts on solutions of Painlev\'e III($D_8$) by formula 
\begin{equation}
z\mapsto z, \qquad w\mapsto w_1=z/w, \qquad p\mapsto p_1=-\frac{w(2wp-1)}{2z}. \label{Backltransf}
\end{equation}
By \eqref{ham} this transformation leads to transformation of $\zeta(z)$.
Variables after transformation we will mark by subscript ${}_1$. 
%Notation ${}_{(1)}$ means that formula is valid no matter of presence of subscript.
We have two useful formulas for the transformation of function $\zeta(z)$
\begin{equation}
\zeta_1=\zeta-pw+1/4, \qquad \zeta'\zeta_1'=z\label{profd}
\end{equation}
which follow from \eqref{ham} and \eqref{conn} respectively.

 In terms of sine-Gordon equation (see Remark \ref{rem:sinh}) B\"acklund transformation is just $v\mapsto -v$. 

\begin{prop}
\label{prop:Okamoto}
\begin{enumerate}
 \item[(i)] Consider a solution $\zeta(z)$ of \eqref{zeta3}, its B\"acklund transformation $\zeta_1(z)$ and 
 functions $\tau(z)$ and $\tau_1(z)$ corresponding to $\zeta(z)$ and $\zeta_1(z)$ by \eqref{tauzeta}.
 Then the functions $\tau(z)$ and $\tau_1(z)$ satisfy equations
 \begin{eqnarray}
  &D_{[\log z]}^2(\tau,\tau_1)-\frac12\left(z\frac{d}{dz}-\frac18\right)(\tau\tau_1)=0 \label{d8eq_1},
  \\
   &D_{[\log z]}^3(\tau,\tau_1)-\frac12\left(z\frac{d}{dz}-\frac18\right)D^1_{[\log z]}(\tau,\tau_1)=0. \label{d8eq_2}
   \end{eqnarray} 
   \item[(ii)] 
   Conversely, consider functions $\tau(z)$ and $\tau_1(z)$ satisfying \eqref{d8eq_1}, \eqref{d8eq_2} 
    and functions $\zeta(z)$ and $\zeta_1(z)$ corresponding to $\tau(z)$ and $\tau_1(z)$ by \eqref{tauzeta} and $\ddot{\zeta}(z)\neq 0$, $\ddot{\zeta_1}(z)\neq 0$.
  Then there exists $D\neq 0$ such that functions $\zeta(z/D), \zeta_1(z/D)$, satisfy equations \eqref{zeta3}
  and $\pi(\zeta(z/D))=\zeta_1(z/D)$.
\end{enumerate}
\end{prop}

Usually bilinear equations in Painlev\'e theory are obtained using infinite order B\"acklund transformation (see e.g. \cite{OKSO}, \cite{OkR}). We follow this approach using $\pi$ of order $2$ for Painlev\'e III($D_8$) that's why
we will call equations \eqref{d8eq_1}, \eqref{d8eq_2} "Okamoto-like equations".

Equation \eqref{d8eq_1} is symmetric under transposition $\tau \leftrightarrow \tau_1$ and equation \eqref{d8eq_2} is skew-symmetric under
this transposition. This is natural since $\pi^2=1$.

\begin{Remark} 
Okamoto-like equations \eqref{d8eq_1}, \eqref{d8eq_2} have symmetry in rescaling $z$, i.e. if
$\tau(z)$, $\tau_1(z)$ is a solution then $\tau(Dz)$, $\tau_1(Dz)$, $D\neq 0$ is also a solution.
In other words any solution of \eqref{d8eq_1}, \eqref{d8eq_2} can be obtained from a Painlev\'e $\tau$ function using such rescaling and a solution with asymptotic behavior \eqref{tauas} corresponds to the value $D=1$
\end{Remark}

\begin{proof} \textit{{(i)}}
 %Connection with $\tau$ functions is 
 %\begin{equation}
 %zH=z\frac{d}{dz}\log \tau, \quad
 %zH_1=z\frac{d}{dz}\log \tau_1.
 %\end{equation}
 %Let us introduce function $A$ given by first equation of \eqref{profd}
 %\begin{equation}
 %A=\zeta-\zeta_1=pw-\frac14.\label{A}
 %\end{equation}
From \eqref{conn} we have 
 \begin{equation}
 z\frac{d}{dz}\zeta=-\frac{z}{w},\qquad z\frac{d}{dz}\zeta_1=-w.\label{temp1}
 \end{equation}
 One could check that simple algebraic identities hold
 \begin{eqnarray}
 \zeta-\zeta_1&=&z\frac{d}{dz}\log \frac{\tau}{\tau_1}=\frac{D^1_{[\log z]}(\tau,\tau_1)}{\tau\tau_1}\label{eq:A},
 \\
 z\frac{d}{dz}(\zeta+\zeta_1)&=&\frac{D^2_{[\log z]}(\tau,\tau_1)}{\tau\tau_1}-\left(\frac{D^1_{[\log z]}(\tau,\tau_1)}{\tau\tau_1}\right)^2,\label{aid2}\\
 \left(z\frac{d}{dz}\right)^2 (\zeta-\zeta_1)&=&\frac{D^3_{[\log z]}(\tau,\tau_1)}{\tau\tau_1}-
 3\frac{D^2_{[\log z]}(\tau,\tau_1)}{\tau\tau_1}\frac{D^1_{[\log z]}(\tau,\tau_1)}{\tau\tau_1}+2\left(\frac{D_{[\log z]}(\tau,\tau_1)}{\tau\tau_1}\right),\label{aid3}
 \end{eqnarray}
 Using \eqref{aid2}, \eqref{profd} and \eqref{temp1} we have
 \begin{equation*}
 \frac{D^2_{[\log z]}(\tau,\tau_1)}{\tau\tau_1}-(\zeta-\zeta_1)^2=z\frac{d}{dz}(\zeta+\zeta_1)=-w-z/w=\zeta-p^2w^2=\zeta-\left(\zeta-\zeta_1+\frac14\right)^2,
 \end{equation*}
 where we used \eqref{temp1} in second equality, \eqref{ham} in third equality and first equation of \eqref{profd} in fourth.
 Using \eqref{eq:A} we finally obtain \eqref{d8eq_1}.
 %\begin{equation}
  % D_{[\log z]}^2(\tau,\tau_1)-\frac12\left(z\frac{d}{dz}-\frac{1}{8}\right)\tau\tau_1=0.\ny
   %\end{equation}
 
 Analogously we could obtain \eqref{d8eq_2}. From \eqref{aid3} we have
 \begin{equation*}
 \begin{aligned}
 \frac{D_{[\log z]}^3(\tau,\tau_1)}{\tau\tau_1}-3\frac{D_{[\log z]}^2(\tau,\tau_1)D^1_{[\log z]}(\tau,\tau_1)}{\tau^2\tau_1^2}
 +2\left(\frac{D^1_{[\log z]}(\tau, \tau_1)}{\tau\tau_1}\right)^3=\left(z\frac{d}{dz}\right)^2(\zeta-\zeta_1)=\\=\left(z\frac{d}{dz}\right)(w-z/w)
 =-z/w+2pw(w+z/w)=z\frac{d}{dz}\zeta-\left(2\zeta-2\zeta_1+\frac12\right)z\frac{d}{dz}(\zeta+\zeta_1)=\\=
 z\frac{d}{dz}\zeta-\left(\frac{D_{[\log z]}^2(\tau,\tau_1)}{\tau\tau_1}-(\zeta-\zeta_1)^2\right)\left(2(\zeta-\zeta_1)+\frac12\right),
 \end{aligned}
 \end{equation*}
 where we use at first \eqref{temp1}, then \eqref{Hameq}, then first equation of \eqref{profd} and \eqref{temp1} and then \eqref{aid2}.
 I.e. we have
 \begin{equation*}
 \frac{D_{[\log z]}^3(\tau,\tau_1)}{\tau\tau_1}-\frac{D_{[\log z]}^2(\tau,\tau_1)D^1_{[\log z]}(\tau,\tau_1)}{\tau^2\tau_1^2}
 =-\frac12\left(\frac{D_{[\log z]}^2(\tau,\tau_1)}{\tau\tau_1}-(\zeta-\zeta_1)^2\right)+z\frac{d}{dz}\zeta.
 \end{equation*}
 Then, using \eqref{d8eq_1} we obtain \eqref{d8eq_2}.
 %Finally one could see that Okamoto-like equations \eqref{d8eq_1}, \eqref{d8eq_2} have symmetry in rescaling $z$, i.e. if
 %$\tau(z)$, $\tau_1(z)$ is a solution then $\tau(Dz)$, $\tau_1(Dz)$, $D\neq 0$ is also a solution. 
 
 %Asymptotic of $\tau$ function which give us solution of \eqref{zeta3} is just \eqref{tauas}. \\
\textit{(ii)} From \eqref{d8eq_1}, \eqref{d8eq_2} we obtain respectively
 \begin{eqnarray}
 \zeta'+\zeta_1'&=&\zeta-\frac14\left(2\zeta-2\zeta_1+\frac12\right)^2\label{prezeta},\\
 \zeta''-\zeta''_1&=&\zeta'-\left(2\zeta-2\zeta_1+\frac12\right)(\zeta'+\zeta_1').
 \end{eqnarray}
 Let us differentiate first equation and then sum and subtract from it second equation. We obtain
 \begin{eqnarray}
 \zeta''&=&-\left(2\zeta-2\zeta_1-\frac12\right)\zeta',\label{aux1}\\
 \zeta_1''&=&\left(2\zeta-2\zeta_1+\frac12\right)\zeta_1' \label{aux_0}.
 \end{eqnarray}
 From these equations we obtain
 \begin{equation*}
 \zeta_1''/
 \zeta_1'=1-\zeta''/\zeta' \Leftrightarrow \zeta'\zeta_1'=Dz,
 \end{equation*}
 where $D\neq 0$ is integration constant (cf. second equation of \eqref{profd}) 
 
 Eliminating $\zeta_1, \zeta_1'$ from \eqref{prezeta} using $\zeta'\zeta_1'=Dz$ and \eqref{aux1} we obtain
 \begin{equation*}
 (\zeta''-\zeta')^2=4(\zeta'^2)(\zeta-\zeta')-4Dz\zeta' \Leftrightarrow (z\ddot\zeta(z))^2=4(\dot\zeta(z))^2(\zeta(z)-z\dot\zeta(z))-4D\dot\zeta(z)
 \end{equation*}
 and evidently we have the same equation on $\zeta_1$.
 From this we see that $\zeta(z/D)$, $\zeta_1(z/D)$ satisfy \eqref{zeta3}.

 Let us check that $\pi(\zeta(z/D))=\zeta_1(z/D)$.
 These functions satisfy second equation of \eqref{profd} from which follows that $\zeta_1(z/D)-\pi(\zeta(z/D))$ is a constant.
Since $\zeta_1(z/D)$ and $\pi(\zeta(z/D))$ satisfy \eqref{zeta3} then $\zeta_1(z/D)\neq \pi(\zeta(z/D))$  iff $\zeta_1(z/D)'=0$ which contradicts \eqref{profd}. 
% 
% Finally we should check that we have not obtained exceptional solution $\zeta(z)=z+1$.
\end{proof}

Now we will present other equations which we call "Toda-like" because they are analogous to similar equations called "Toda equations" in \cite{OKSO}.
 
\begin{prop}
\begin{enumerate}
\item[(i)] 
Let $\zeta(z)$ denotes a solution of \eqref{zeta3}, $\zeta_1(z)$ denotes its B\"acklund transformation and 
 functions $\tau(z)$ and $\tau_1(z)$ correspond to $\zeta(z)$ and $\zeta_1(z)$ by \eqref{tauzeta}.
 Then the functions $\tau(z)$ and $\tau_1(z)$ satisfy equations
 \begin{equation}
  \begin{aligned}
  D^2_{[\log z]}(\tau,\tau)&=2Cz^{1/2}\tau_1^2, \\
   D^2_{[\log z]}(\tau_1,\tau_1)&=2C^{-1}z^{1/2}\tau^2. \label{eq:Toda}
   \end{aligned}
   \end{equation}
\item[(ii)] 
Consider functions $\tau(z)$ and $\tau_1(z)$ satisfy \eqref{eq:Toda} 
and functions $\zeta(z)$ and $\zeta_1(z)$ corresponding to $\tau(z)$ and $\tau_1(z)$ by \eqref{tauzeta} and $\ddot{\zeta}(z)\neq 0$, $\ddot{\zeta_1}(z)\neq 0$.
Then there exists $K$ such that functions $\zeta(z)-K, \zeta_1(z)-K$ satisfy equations \eqref{zeta3}
  and $\pi(\zeta(z)-K)=\zeta_1(z)-K$.
\end{enumerate}
  Constant $C$ in \eqref{eq:Toda} depends on normalization of $\tau$ and $\tau_1$.
   \label{prop:Toda}
   \end{prop}
%\begin{Remark}
%Changing normalization of $\tau$ or $\tau_1$ we can arbitrary change constant $C$.
%\end{Remark}
\begin{Remark}
Usually Toda-like equations are written in logarithmic form, i.e. \eqref{eq:Toda} could be rewritten as
\begin{eqnarray}
  \left(z\frac{d}{dz}\right)^2\log\tau=Cz^{1/2}\tau_1^2/\tau^2, \quad
   \left(z\frac{d}{dz}\right)^2\log\tau_1=C^{-1}z^{1/2}\tau^2/\tau_1^2 \label{eq:Todaln}
   \end{eqnarray}
\end{Remark}
\begin{Remark} \label{fixrem}
Toda-like equations \eqref{eq:Toda} have symmetry in multiplying by $z^K$, i.e. if
$\tau(z)$, $\tau_1(z)$ is a solution then $z^K\tau(z)$, $z^K\tau_1(z)$ for any $K$ is also a solution. In other words, any solution of \eqref{eq:Toda}
can be obtained from the Painlev\'e $\tau$ function by this multiplication and solution with asymptotic behavior \eqref{tauas} corresponds to the value $K=0$. 
\end{Remark}
\begin{proof}
\textit{(i)}
Using \eqref{conn}, \eqref{profd} and \eqref{eq:A} we have
 \begin{eqnarray*}
z\frac{d}{dz}\left(\log\frac{\tau}{\tau_1}\right)=\zeta-\zeta_1&=&pw-1/4=-\frac{z\ddot\zeta}{2\dot\zeta}-1/4=-\frac{\zeta''}{2\zeta'}+1/4\Leftrightarrow\\
  \zeta'z^{-1/2}&=&C\frac{\tau_1^2}{\tau^2}\Leftrightarrow (\log\tau)''=Cz^{1/2}\frac{\tau_1^2}{\tau^2}
 \end{eqnarray*} 
 from which we have Toda-like equation with constant $C$ and its symmetric one with constant $C_1$
 \begin{equation*}
 \begin{aligned}
 D^2_{[\log z]}(\tau,\tau)=2C\tau_1^2 z^{1/2}, \qquad 
  D^2_{[\log z]}(\tau_1,\tau_1)=2C_1\tau^2 z^{1/2}. 
 \end{aligned}
 \end{equation*}
 Multiplying these equations by each other we have  $\zeta'\zeta_1'=CC_1z$. But from second equation of \eqref{profd} we have $CC_1=1$.
 %Finally one could see that Toda-like equations \eqref{eq:Toda} have symmetry in multiplying on $z^K$, i.e. if
 %$\tau(z)$, $\tau_1(z)$ is a solution then $z^K\tau(z)$, $z^K\tau_1(z)$ for any $K$ is also a solution.
 \\
\textit{(ii)}  Multiplying first and second equations of \eqref{eq:Todaln} we have  $\zeta'\zeta_1'=z$. Acting by $\left(z\frac{d}{dz}\right)^2$ on the second equation of \eqref{eq:Todaln}
 we have 
 \begin{equation*}
-\frac{\zeta'''}{2\zeta'}+\frac{\zeta''^{2}}{2\zeta'^2}=\zeta'-\zeta_1'.
 \end{equation*}
 Substituting $\zeta_1'=z/\zeta'$ we have
 \begin{equation}
 -\zeta'''\zeta'+\zeta''^2=2\zeta'^3-2z\zeta'. \label{eq:zetac}
 \end{equation}
 Let us denote
 \begin{equation*}
  f(z)=z^2\ddot\zeta(z)^2-4\dot\zeta(z)^2(\zeta(z)-z\dot\zeta(z))+4\dot\zeta(z)=\frac{1}{z^2}((\zeta''-\zeta')^2-4\zeta'^2(\zeta-\zeta')+4z\zeta').
  \end{equation*}
  %We want to proof that \eqref{zeta3} holds, i.e. $f(z)=0$.
 Differentiating expression for $f(z)$ we have
 \begin{equation*}
 \frac{z^2}{2(\zeta''-\zeta')}f'=\zeta'''-2\zeta''+\zeta'+6\zeta'^2-4\zeta\zeta'+2z.
 \end{equation*}
 Then equation \eqref{eq:zetac} can be rewritten as
 \begin{equation*}
 z^2f=\frac{z^2\zeta'}{2(\zeta''-\zeta')}f' \Leftrightarrow 2\ddot\zeta f=\dot\zeta \dot f \Leftrightarrow f=4K\dot\zeta^2,
 \end{equation*}
 where $K$ is some integration constant.
 Evidently we have the same equation on $\zeta_1$ with integration constant $K_1$.
 Then functions $\zeta(z)-K$ and $\zeta_1-K_1$ are solutions of \eqref{zeta3}.
 
 Similar to the proof of Proposition \ref{prop:Okamoto} we get $\pi(\zeta-K)=\zeta_1-K_1$, because
 $\zeta'\zeta_1'=z$. Similar to the proof of Proposition \ref{prop:Okamoto} we see that $\zeta(z)\neq z+1+K$. In remains to prove that $K_1=K$, this follows from the fact that both $\pi(\tau)$ and $\tau_1$ satisfy first equation of \eqref{eq:Toda}. 
 \end{proof}

 % We understood above that between solutions of \eqref{zeta3} and \eqref{piii} exists one-to-one correspondence.

In the Section \ref{sec:FxNSR} we will see appearance of Toda-like and Okamoto-like equations in the framework
of representation theory of Virasoro algebra.
 
\section{Painlev\'e $\tau$ function and Virasoro algebra}
\label{sec:Kiev}
\subsection{Virasoro conformal blocks}
\label{ssec:Kiev}
The Virasoro algebra (which we denote by $\Vir$) is generated by  $L_n$, $n\in\mathbb{Z}$ with relations
\begin{equation}
[L_n,L_m]=(n-m)L_{n+m}+\frac{n^3-n}{12}c\delta_{n+m,0}. 
\end{equation}
Here $c$ is an additional central generator, which acts on representations below as multiplication by a complex number. Therefore we consider $c$ as a complex number, which we call central charge.

Denote the Verma module of $\Vir$ by $\pi^{\D}_{\Vir}$. This module is generated by a highest weight vector $|\D\ra$
\begin{equation}
L_0|\D\ra=\D|\D\ra, \quad L_n|\D\ra=0, \quad n>0,
\end{equation}
where $\D \in \mathbb{C}$ is called the weight of $|\D\ra$.
Here and below in analogous situations the representation space is freely spanned by vectors obtained by the action of the operators
$L_{-n}, \,n>0$ on $|\D\ra$, i.e. vectors of the form $L_{-n_1}L_{-n_2}\cdot L_{-n_k}|\Delta\rangle$ for $n_1\geq n_2 \geq \ldots n_k$ form a basis in the Verma module $\pi^{\D}_{\Vir}$. Actually the Verma module is determined by the pair of complex numbers $\Delta$ and $c$, but we just write $\pi^{\D}_{\Vir}$ if a fixed value of $c$ has been chosen.

We use everywhere in this work not hermitian, but a complex symmetric scalar product.
We define scalar product on $\pi^{\D}_{\Vir}$ by conjugation $L_n^+=L_{-n}$ and normalization $\la\D|\D\ra=1$.
This scalar product is called Shapovalov form.
Note that in this paper we will normalize all highest weight vectors of $\Vir$ Verma modules on $1$ unless otherwise stated.

We will say that $\Delta$ is generic if $\Delta \neq \Delta_{m,n}=((b^{-1}+b)^2-(mb^{-1}+nb)^2)/4$, where parameter $b$ is defined by the equation $c=1+6(b^{-1}+b)^2$ and $m,n \in \mathbb{Z}_{> 0}$. For such generic $\Delta$ Shapovalov form on Verma module $\pi^{\D}_{\Vir}$ is nondegenerate \cite[Theorem 4.2]{IK}. Therefore for generic $\Delta$ Verma module $\pi^{\D}_{\Vir}$ is irreducible.

%It is useful to parametrize basis vectors on Verma module by Young diagrams such that basis vectors equal $L_{-\lmb}|\varnothing\ra$, where $\lmb$ is Young 
%diagram. By $L_{-\lmb}$ we denote product $L_{-\lmb_1}L_{-\lmb_2}\ldots,\, \lmb_1\leq \lmb_2 \leq\ldots$. 

Let us define so-called irregular limit of conformal block.
At first we define Whittaker vector $|W(z)\ra$ by formula
\begin{equation} \label{eq:WhitVir}
|W(z)\ra=z^{\D}\sum\limits_{N=0}^{\infty}z^{N}|N\ra,\qquad |N\ra\in\pi^{\D}_{\Vir} , \qquad L_0|N\ra=(\D+N)|N\ra,
\end{equation}
where
\begin{equation}
L_1|N\ra=|N-1\ra, \quad N>0, \quad 
L_k|N\ra=0, k>1 ,
\end{equation}
or equivalently,
\begin{equation}
L_1|W(z)\ra=z|W(z)\ra,\quad L_k|W(z)\ra=0, k>1.
\end{equation}
Using these conditions and the normalization $|0\ra=|\D\ra$ one can compute vectors all $|N\rangle$  inductively. Moreover, it is easy to see that if Shapovalov form is nondegenerate then such system of vectors exists and unique. Therefore for generic values of $\Delta$ the Whittaker vector is well defined.  

Note that it is enough to impose $L_1$ and $L_2$ relations since the action of the other $L_k$, $k>2$ follows from the Virasoro commutation relations.
Note also that this definition of $|W(z)\ra$ slightly differs from the one used in \cite{KGP}, our $z$ is $z^{1/2}$ in loc. cit.

The irregular (or Whittaker, or Gaiotto) limit of conformal block is defined by
\begin{equation}\label{eq:Whitcom}
\kr{F}_c(\D|z)=\la W(1)|W(z)\ra=z^{\D}\sum\limits_{N=0}^{\infty}z^{N}\la N|N\ra.
\end{equation}

We are ready to formulate fact, conjectured in \cite{GIL1302}  (which is based on a result of work \cite{GIL1207}) and proved in different ways in \cite{KGP} and \cite{ILT}.
\begin{thm}
\label{thm:1}
The expansion of the Painlev\'e III($D_8$) $\tau$ function near $z=0$ can be written as 
\begin{equation}\label{Kiev}
\tau(\sg,s|z)\propto\sum_{n\in\mathbb{Z}}C(\sg+n)s^n \kr{F}((\sg+n)^2|z), \quad \operatorname{Re}\in\mathbb{R}\setminus\left\{\frac12\mathbb{Z}\right\}, \quad s\in\mathbb{C}\setminus\{0\},
\end{equation}
where
$\kr{F}(\sg^2|z)=\kr{F}_{c=1}(\sg^2|z)$. The coefficients $C(\sg)$ are defined by
$ C(\sg)=1/\Bigl(\G(1-2\sg)\G(1+2\sg)\Bigr),$
where $\G(\cdot)$ is the Barnes $\G$-function. The parameters $s$ and $\sg$ in \eqref{Kiev} are integration constants of  equation~\eqref{zeta3}.
Notation $\propto$ means constant proportionality.
\end{thm}
\begin{Remark}
Actually parameter $\sigma$ in formula \eqref{Kiev} can be $\sigma \in \mathbb{C}\setminus\left\{\frac12\mathbb{Z}\right\}$. We use smaller region in order to compare these $\tau$ functions with the ones from the family given in Proposition \ref{Jimboas}. 

In fact, formula \eqref{Kiev} gives us full power series whose first terms  one sees in the asymptotic behavior \eqref{tauas} (see the next Subsection for details).
Parameters $\sg$ in these formulas are the same and connection between $s$ and $\td{s}$ is given in the next Subsection.
\end{Remark}

\subsection{B\"acklund transformation in terms of $\tau$ functions}
\label{ssec:BTParam}
B\"acklund transformation acts on $\tau$ function due to the formulas \eqref{ham} and \eqref{tauzeta}.
We want to understand its relation with \eqref{Kiev}. We will see that action of B\"acklund transformation is just action on parameters $\sg,\td{s}$.
 
The $\tau$ function \eqref{Kiev} has obvious symmetries
\begin{equation}
\tau(\sg,s|z)=s\tau(\sg+1,s|z), \quad \tau(\sg,s|z)=\tau(-\sg,s^{-1}|z). \label{obvsymm}
\end{equation}
Therefore we can move $\operatorname{Re} \sg$ into interval $(0,\frac12)$.

 Let us mention the fact, that coefficients $C(\sg+n)$ in 
\eqref{Kiev} could be made rational functions. Indeed
\begin{equation}
\frac{C(\sg+n)}{C(\sg)}=\left(\frac{\Gamma(-2\sg)}{\Gamma(2\sg)}\right)^{2n}
  \left\{\begin{aligned}                  
	  \frac{(-1)^{\lfloor n\rfloor}}{(2\sg)^{2n}\prod\limits_{i=1}^{2n-1}(2\sg+i)^{2(2n-i)}}, \quad n\geq 0
  \\  
	  \frac{(-1)^{\lfloor -n\rfloor}}{(2\sg)^{-2n}\prod\limits_{i=1}^{-2n-1}(-2\sg+i)^{2(-2n-i)}}, \quad n<0
  \end{aligned}\right.
  =\left(\frac{\Gamma(-2\sg)}{\Gamma(2\sg)}\right)^{2n}\td{C}(\sg,n),\label{rationalC}
\end{equation}
  where $2n\in\mathbb{Z}$.
  So we can introduce $\td{s}(s,\sg)=s\left(\frac{\Gamma(-2\sg)}{\Gamma(2\sg)}\right)^{2}$ and consider
  \[\tau(\sg,\td{s}|z)\propto\sum_{n\in\mathbb{Z}}\td{C}(\sg,n)\td{s}^n \kr{F}((\sg+n)^2|z)\propto\tau(\sg,s|z).\]

 \begin{prop}
 B\"acklund transformation of $\tau$ function is given by formula
 $\tau_1(\sg,s|z)=\pi(\tau(\sg,s|z))\propto\tau(1/2-\sg,s^{-1}|z)$.\label{taubackl}
 \end{prop}
 Now we give straightforward proof based on Proposition \ref{Jimboas} This proof do not use formula \eqref{Kiev} and actually can be given in Section \ref{sec:Painleve}.  Below in Subsection \ref{ssec:Todaobt} we will give another proof using representation theory.
 \begin{proof}
 Let us calculate the asymptotic behavior of \eqref{Kiev} in case $0<\operatorname{Re} \sg<\frac12, s\in\mathbb{C}\setminus \{0\}$
 \begin{equation*}
  \tau(\sg,s|z)\propto z^{\sg^2}\left(1+\frac{z}{2\sg^2}\right)+\td{C}(\sg,-1)\td{s}^{-1}z^{(\sg-1)^2}+o(|z|^{\sg^2+1}).
%   z^{\sg^2}\left(1+\frac{z}{2\sg^2}-\frac{\td{s}^{-1}z^{1-2\sg}}{(2\sg(1-2\sg))^2}+o(|z|)\right), \label{tauas2}
  \end{equation*}
  This is just the asymptotic behavior \eqref{tauas} with coinciding parameters $\sg$ and $\td{s}$.
  %Due to the Proposition \ref{Jimboas} in this domain of $\sg,s$ there is no identical $\tau$ functions from the family $\tau(\sg,s|z)$. 
  %and these $\tau$ functions 
  %are $\tau$ functions of general position.
 Let us then calculate the asymptotic behavior for $w(z)$ (cf. \eqref{was})
 \begin{eqnarray*}
 \zeta(z)=z\frac{d}{dz}\left(\sg^2 \log z+\frac{z}{2\sg^2}-\frac{\td{s}^{-1}z^{1-2\sg}}{(2\sg(1-2\sg))^2}
 -\frac12\frac{\td{s}^{-2}z^{2-4\sg}}{(2\sg(1-2\sg))^4}+o(|z|)\right)=\ny\\=
 \sg^2+\frac{z}{2\sg^2}-\frac{\td{s}^{-1}z^{1-2\sg}}{(2\sg)^2(1-2\sg)}-\frac{\td{s}^{-2}z^{2-4\sg}}{(2\sg)^4(1-2\sg)^3}+o(|z|);\\
 w(z)=-\frac{1}{\dot\zeta(z)}=-\frac{1}{\frac{1}{2\sg^2}-\frac{\td{s}^{-1}z^{-2\sg}}{(2\sg)^2}-\frac{\td{s}^{-2}z^{1-4\sg}}{(2\sg)^4(1-2\sg)^2}+o(1)}=4\sg^2\td{s}z^{2\sg}(1+o(1)).
 \end{eqnarray*}
 After the B\"acklund transformation $w\mapsto z/w$ the asymptotic will be
$\frac{z^{1-2\sg}}{4\sg^2\td{s}}(1+o(1))$.
%$4\sg_1^{2}\td{s}_1z^{2\sg_1}(1+o(1))$.
%From this asymptotics due to the Proposition \ref{Jimboas} 
Therefore the function $\pi(\tau(\sg,s|z))$ belongs to the same two-parametric family due to the Proposition \ref{Jimboas} and the corresponding
parameters $\sg_1, \td{s}_1$ are given by
\begin{equation*}
\sg_1=1/2-\sg,\quad
\td{s}_1=\frac{1}{(2\sg)^2(1-2\sg)^2\td{s}}\Rightarrow s_1=s^{-1}. 
\end{equation*}
So finally
\begin{equation}
\tau(\sg,s|z)\mapsto\tau_1(\sg,s|z)\propto\tau(1/2-\sg,s^{-1}|z)=\tau(\sg-1/2,s|z) \label{halfshift}. 
\end{equation}
\end{proof}
\begin{Remark}
The power series decomposition of $\tau_1$ can be written as
\begin{equation}\label{Kiev1}
\tau_1(\sg,s|z)\propto\sum_{n\in\mathbb{Z}+1/2}C(\sg+n)s^n \kr{F}((\sg+n)^2|z), \quad \operatorname{Re}\sg\in\mathbb{R}\setminus\left\{\frac12\mathbb{Z}\right\}, \quad s\in\mathbb{C}\setminus\{0\},
\end{equation}
which differs from \eqref{Kiev} only in the region of summation.
\end{Remark}
Using Proposition \ref{taubackl} and Propositions  \ref{prop:Okamoto}, \ref{prop:Toda} we can see that $\tau(\sg,s|z)$ given by the right side of \eqref{Kiev}
satisfies equations on function $\tau(\sg|z)$ which follow from Toda-like and Okamoto-like equations as equations in $\sg$ and $z$. 
These equations will be differential on $z$ and difference on $\sg$. For instance
Okamoto-like equation turns to
\begin{eqnarray}
  &D_{[\log z]}^2(\tau(\sg|z),\tau(\sg-1/2|z))-\frac12\left(z\frac{d}{dz}-\frac18\right)(\tau(\sg|z)\tau(\sg-1/2|z))=0 \label{difd8eq_1},
  \\
   &D_{[\log z]}^3(\tau(\sg|z),\tau(\sg-1/2|z))-\frac12\left(z\frac{d}{dz}-\frac18\right)D^1_{[\log z]}(\tau(\sg|z),\tau(\sg-1/2|z))=0 \label{difd8eq_2}.
   \end{eqnarray}
Analogous form of the Toda-like equation is given in next Remark. 
\begin{Remark}
\label{Crem}
We could directly determine constant $C$ in Toda-like equations \eqref{eq:Toda} for
$\tau=\tau(\sg,s|z)$, $\tau_1=\tau(\sg-1/2,s|z)$ where normalization of $\tau(\sg,s|z)$ is given by equality in \eqref{Kiev}.
Indeed the first equation of \eqref{eq:Toda} could be rewritten as
\begin{equation*}
\zeta'=Cz^{1/2}\frac{\tau(1/2-\sg,s^{-1}|z)^2}{\tau(\sg,s|z)^2}
\end{equation*}
with the asymptotic behavior of l.h.s. and r.h.s.
\begin{equation*}
\zeta'=-\frac{z^{1-2\sg}}{4\sg^2 \td{s}}(1+o(1)),\qquad
r.h.s.=Cz^{1/2}\left(\frac{C(1/2-\sg)}{C(\sg)}\right)^2(1+o(1)) z^{1/2-2\sg} \Rightarrow 
C=-s^{-1}.
\end{equation*}
So the $\tau$ function given by decomposition \eqref{Kiev} satisfies differential-difference equation
\begin{equation}
1/2D^2_{[\log z]}(\tau(\sg|z), \tau(\sg|z))=-s^{-1}z^{1/2}\tau(\sg-1/2)^2,\label{difToda2}
\end{equation}
or using first relation from \eqref{obvsymm}
\begin{equation}
1/2D^2_{[\log z]}(\tau(\sg|z), \tau(\sg|z))=-z^{1/2} \tau(\sg+1/2|z)\tau(\sg-1/2|z).\label{difToda} 
\end{equation}

Analogously for normalization of $\tau_1$ given  by equality in	 \eqref{Kiev1} we have $C=-1$.
\end{Remark}
We will use the differential-difference equations  \eqref{difd8eq_1}, \eqref{difd8eq_2}, \eqref{difToda} in Section \ref{sec:FxNSR}.

\subsection{Algebraic solution of Painlev\'e III($D_8$) equation}
\label{ssec:algsol}
There exist only two rational solutions of Painlev\'e III($D_8$) equation: $w(z)=\pm \sqrt{z}$ (\cite{Gromak}).
These solutions are only invariant solutions under the B\"acklund transformation $\pi(\pm \sqrt{z})=\pm \sqrt{z}$.
Using \eqref{conn} we obtain
\begin{equation*}
\dot{\zeta}(z)=\mp\frac1{\sqrt{z}} \Rightarrow p(z)=\pm \frac1{4\sqrt{z}}.
\end{equation*}
Therefore, using \eqref{ham} and \eqref{tauzeta} we find $\zeta$ and $\tau$ functions 
\begin{equation}
\zeta(z)=1/16 \mp 2 \sqrt{z},\qquad  \tau(z)\propto z^{1/16}e^{\mp 4\sqrt{z}}.\label{taualg}
\end{equation}
This formula is also more or less known, for example for the upper sign it follows (as well as next formula \eqref{eq:tau=exp}) from \cite[eq. (3.52)]{BGT} after substitution $\sigma=1/4$ and probably from other sources.
%This result up to $z^{1/16}$ (see Remark \ref{rem:Todaamb}) could be obtained from Toda-like equations \eqref{eq:Toda}. Really
%we have that $\tau(z)\propto\tau_1(z)$ in B\"acklund invariant point, so \eqref{eq:Toda} turns to
%\begin{equation}
%\left(\frac{\tau'(z)}{\tau(z)}\right)'=\mp z^{1/2}.
%\end{equation}
%Finally we have that $\tau(z)=z^{1/16-K/4}e^{\mp 4\sqrt{z}}$, where $K$ is the same as in proof of Proposition \ref{prop:Toda}.

On the other hand from Proposition \ref{taubackl} follows that there are two B\"acklund invariant solutions $\tau(1/4,\pm 1|z)$ given by the right side of \eqref{Kiev}. Comparing first terms of the power series expansion
of \eqref{Kiev} and the expression \eqref{taualg} for $\tau$ function we obtain
\begin{equation}\label{eq:tau=exp}
\tau(1/4,\pm 1|z)=C(1/4) z^{1/16}e^{\mp4\sqrt{z}}. 
\end{equation}
So, using \eqref{rationalC}, we obtain following relation on conformal blocks
\begin{equation}
\sum_{n\in\mathbb{Z}}(\mp 1)^n B_n\mathcal{F}((1/4+n)^2|z)=z^{1/16}e^{\mp4\sqrt{z}},
\label{blockquarter}
\end{equation}
where coefficients $B_n$ are equal to
\begin{equation}
	B_n=\frac{2^{4n^2+2n}}{\prod\nolimits_{i=0}^{2n-1}\left(2i+1\right)^{2(2n-i)}}, \, n\geq 0, 
	\qquad
	B_n=\frac{2^{4n^2+2n}}{\prod\nolimits_{i=0}^{-2n-2}\left(2i+1\right)^{2(-2n-i-1)}}, \, n<0\label{Bn}.
\end{equation}

In the remaining part of this Subsection we will prove relation \eqref{blockquarter} using representation theory.

Introduce Heisenberg algebra with generators $a_r,\, r\in\mathbb{Z}+1/2$ and relations $[a_r,a_s]=r\delta_{r+s,0}$.
Consider Fock module $\mathrm{F}$ generated by highest weight vector $|\varnothing\ra$ which satisfies $a_r|\varnothing\ra=0,\, r>0$.
Then one can introduce action of the algebra $\Vir$ with $c=1$ by formula
\begin{equation}
L_n=\frac12\sum\limits_{r\in\mathbb{Z}+1/2} :a_{n-r}a_r:+\frac{1}{16}\delta_{n,0}, \label{La}
\end{equation}
where $:\ldots:$ is standard Heisenberg normal ordering. Our first goal is to describe $\mathrm{F}$ as a $\Vir$ module. This module cannot be irreducible since Heisenberg algebra has half-integer indices but $\Vir$ has only integer indices.

%Then one can see that vector $|\varnothing\ra$ is highest weight vector with respect to $\Vir$ algebra with highest weight $1/16$.
%So vector $|\varnothing\ra$ generate $\Vir$ Verma module by acting of $L_{-n}, n>0$. No relations between vectors of this module hold because due to
%Feigin-Fuks theorem (see \cite{IK}) module $\pi^{1/16}_{\Vir}$ is irreducible. Below we see that there is another highest weight vectors in 
%$\mathrm{F}$ with respect to $\Vir$ which also generate irreducible Verma modules. Vectors from $\Vir$ Verma modules with different highest weight
%will be orthogonal as usual.

On $\mathrm{F}$ we have scalar product defined by $a_r^+=a_{-r},\, \la\varnothing|\varnothing\ra=1$. This product is nondegenerate.
Due to \eqref{La} we have $L_n^+=L_{-n}$. Therefore this scalar product coincides with Shapovalov form on $\Vir$ submodules on $\mathrm{F}$.

Since $[L_0,a_r]=-ra_r$ and $L_0|\varnothing\ra=\frac{1}{16}|\varnothing\ra$ character of $\mathrm{F}$ equals
\begin{equation*}
\mathrm{ch}(\mathrm{F})=\mathrm{Tr}(z^{L_0})=\frac{z^{1/16}}{\prod\nolimits_{r=0}^{\infty}(1-z^{r+1/2})}. 
\end{equation*}

This formula can be rewritten using Gauss relation 
\begin{equation}
\sum\limits_{k=0}^{\infty}z^{k(k+1)/2}=\prod\limits_{k=1}^{\infty}\frac{1-z^{2k}}{1-z^{2k-1}} \label{Gauss}
\end{equation}
which follows from Jacobi triple product identity
\begin{equation}
\prod_{k=1}^{\infty}(1-z^{2k})(1+z^{2k-1}y^2)(1+z^{2k-1}y^{-2})=\sum_{k=-\infty}^{\infty}z^{k^2}y^{2k} \label{JTP}
\end{equation}
after the substitution $y\mapsto z$ and then $z\mapsto \sqrt{z}$.
%\begin{equation}
%\sum_{k=0}^{\infty} z^{k(k+1)/2}=\prod_{k=1}^{\infty}\frac{(1-z^{2k})^2}{(1-z^k)}=\prod_{k=1}^{\infty}\frac{1-z^{2k}}{1-z^{2k-1}}
%\end{equation}

Now using \eqref{Gauss} we get
\begin{equation}
\mathrm{ch}(\mathrm{F})=\sum_{n\in\mathbb{Z}}\frac{z^{(n+1/4)^2}}{\prod\nolimits_{k=1}^{\infty}(1-z^k)}.\label{Fchardecomp}
\end{equation}
In the right side of \eqref{Fchardecomp} we have sum of characters of Verma modules $\pi_{\Vir}^{(n+1/4)^2}$. Moreover we have
\begin{prop} \label{prop:Fockdecomp}
Fock module $\mathrm{F}$ is isomorphic to the direct sum of $\Vir$ $c=1$ Verma modules with highest weight $(n+1/4)^2,\, n\in\mathbb{Z}$
\begin{equation}
\mathrm{F}\cong\bigoplus_{n\in\mathbb{Z}}\pi_{\Vir}^{(n+1/4)^2} \label{Fockdecomp}
\end{equation}
\end{prop}
%This decomposition is well known and standard (see \cite{Zambook}, \cite{Zam:1987}).

\begin{proof} It is sufficient to prove the existence of vectors $|(n+1/4)^2\rangle$ such that 
\begin{equation} \label{eq:n+1/4}
L_k |(n+1/4)^2\rangle=0,\text{ for } k>0,\qquad L_0 |(n+1/4)^2\rangle=(n+1/4)^2|(n+1/4)^2\rangle
\end{equation}
and submodules generated by $|(n+1/4)^2\rangle$ are orthogonal with respect to Shapovalov form on $\mathrm{F}.$ Indeed the Verma modules on Virasoro algebra with $c=1$ and $\Delta=(n+1/4)^2$ are irreducible (since correponding $\Delta$ are generic), therefore the vectors $|(n+1/4)^2\rangle$ generate Verma modules $\pi_{\Vir}^{(n+1/4)^2}$. These modules are linear independent due to orthogonality and their direct sum is isomorphic to $\mathrm{F}$ due to character identity~\eqref{Fockdecomp}. 
	
Existence of vectors $|(n+1/4)^2\rangle$ follows by induction. As base we choose $|(0+1/4)^2\rangle=|\varnothing\rangle$. Assume the existence of vectors $|(n+1/4)^2\rangle$, $-m<n<m$, then we can define $|(-m+1/4)^2\rangle$ as the highest weight vector of orthogonal complement $\left(\oplus_{n=-m+1}^{m-1}\pi^{(n+1/4)^2}_{\Vir}\right)^\perp$. The conditions \eqref{eq:n+1/4} will be satisfied due to orthogonality and character identity \eqref{Fchardecomp}. Then we define $|(m+1/4)^2\rangle$ as the highest weight vector of orthogonal complement $\left(\oplus_{n=-m}^{m-1}\pi^{(n+1/4)^2}_{\Vir}\right)^\perp$ and this finishes the step of induction. \end{proof}

The decomposition \eqref{Fockdecomp} was stated in \cite{Zam:1987},\cite{Zambook} (and probably in other sources). We will deduce formula \eqref{blockquarter} from this decomposition. We will need explicit formulas for vectors $|(n+1/4)^2\rangle$, these formulas are given in e.g. \cite{Carlsson:2015}. First we should recall so-called boson-fermion correspondence (see e.g.~\cite{KacRaina}). 

Let us extend the Heisenberg algebra by generators $a_n$, $n\in\mathbb{Z}$, $[a_n,a_m]=n\delta_{n+m,0}$ and $[a_n,a_r]=0$, $r \in \mathbb{Z}+1/2$.  Consider also corresponding extension $\overline{\mathrm{F}}(k)$ of Fock module $\mathrm{F}$ with
$a_0$ acting on this module as number $k/\sqrt{2}$, $k \in \mathbb{Z}$. We will denote corresponding vacua vectors by $|\varnothing,k\rangle$. Consider then the sum $\oplus_{k\in\mathbb{Z}}\overline{\mathrm{F}}(k)$ and denote by $S$ the operator $S\colon \overline{\mathrm{F}}(k) \rightarrow \overline{\mathrm{F}}(k+1)$
determined by formulas
\[ S |\varnothing,k\rangle=|\varnothing,k+1\rangle,\qquad [S,a_n]=0,\text{ for } n\neq 0.
\]
Introduce operators $\td{\psi}_r,\td{\psi}^*_r$, $r \in \frac12\mathbb{Z}+\frac14$ by formulas
\begin{equation}
\begin{aligned}
\td{\psi}(z)&=
\!\!\!\sum_{r\in\frac12\mathbb{Z}+\frac14}\td{\psi}_r z^{-r-1/4}=
Sz^{a_0/\sqrt{2}}
\exp\left(\sum\limits_{2j\in \mathbb{Z}_{> 0}}\frac{z^j}{j\sqrt{2}}a_{-j}\right)\exp\left(\sum\limits_{2j\in \mathbb{Z}_{>0}} \frac{z^{-j}}{-j\sqrt{2}}a_{j}\right),\\
\td{\psi}^*(z)&=\!\!\!\sum_{r\in\frac12\mathbb{Z}+\frac14}
\td{\psi}^*_r z^{-r-1/4}=
S^{-1}z^{-a_0/\sqrt{2}}\exp\left(\sum\limits\limits_{2j\in \mathbb{Z}_{>0}} \frac{z^{j}}{-j\sqrt{2}}a_{-j}\right)\exp\left(\sum\limits\limits_{2j\in \mathbb{Z}_{>0}} \frac{z^{-j}}{j\sqrt{2}}a_{j}\right).\label{bfc}
\end{aligned}
\end{equation}

%Consider infinite-dimensional linear space $V$ spanned by basis vectors $v_{m/2}, m\in\mathbb{Z}+\frac12$
%Next let us define space $\Lambda^{\frac{\infty}{2}}V$ spanned by infinite wedge product $v_{i_1}\wedge v_{i_2} \ldots$ as basis vectors such that
%$i_k<i_{k+1}$ and $i_{k+1}=i_{k}+1/2$ for $k\gg 0$. Introduce free-fermion generators $\td{\psi}_r$ and $\td{\psi}^*_r$ with indexes $2r\in\mathbb{Z}+
%\frac12$ such that $\{\td{\psi}_r,\td{\psi}_s\}=\delta_{r+s,0}$, $\{\td{\psi}^*_r,\td{\psi}^*_s\}=\delta_{r+s,0}$, $\{\td{\psi}^*_r,\td{\psi}_s\}=0$ and vacuum vector
%$\widetilde{|\varnothing,0\ra}=v_{1/4}\wedge v_{3/4}\wedge v_{5/4}\wedge\ldots$. Space $\Lambda^{\frac{\infty}{2}}V$ could be
%considered as module generated by acting of these fermions on $\widetilde{|\varnothing\ra}$ with conditions $\td{\psi}_r \widetilde{|\varnothing\ra}=0,\,
%\td{\psi}^*_r \widetilde{|\varnothing\ra}=0,\, n>0$. Precise action of $\td{\psi}_r$, $\td{\psi}_r^*$ on above wedge realization of 
%$\Lambda^{\frac{\infty}{2}}V$ is given by
%\begin{align*}
%\td{\psi}_r (v_{i_1}\wedge v_{i_2} \ldots)=v_{r}\wedge v_{i_1} \wedge v_{i_2}, \ldots\\
%\td{\psi}^*_r (v_{-r} \wedge v_{i_1} \wedge v_{i_2} \ldots)=v_{i_1} \wedge v_{i_2} \ldots
%\end{align*}
%In second formula we have in mind that action is zero when thero is no $v_{-r}$ in wedge product.
%It useful to introduce also vectors $\widetilde{|\varnothing,k\ra}=v_{k/2+1/4}\wedge v_{k/2+3/4}\wedge v_{k/2+5/4}\wedge\ldots$
The boson-fermion correspondence can be stated as follows
\begin{prop}
Operators $\td{\psi}_r, \td{\psi}_r^*$ satisfy Clifford algebra relations:
\[\{\td{\psi}_r,\td{\psi}_s\}=0,\quad \{\td{\psi}_r,\td{\psi}^*_s\}=\delta_{r+s,0},\quad \{\td{\psi}^*_r,\td{\psi}_s^*\}=0. \]
The space $\oplus_{k\in\mathbb{Z}}\overline{\mathrm{F}}(k)$ is a Fock representation of this Clifford algebra generated by vector $|\varnothing,0\rangle$ such that 
\[\td{\psi}_r|\varnothing,0\rangle=\td{\psi}_r^*|\varnothing,0\rangle=0, \quad \text{ for }r>0.\]
The Heisenberg algebra generators $a_j$, $j\in \frac12\mathbb{Z}$ in terms of $\td{\psi}_r, \td{\psi}_r^*$ are given by the formula
\begin{equation}\label{eq:a:psi}
a_{j}=\frac{1}{\sqrt{2}}\sum_{r\in\frac12\mathbb{Z}+\frac14}:\!\td{\psi}_{r}\td{\psi}^*_{j-r}\!:,
\end{equation}
where usual fermion normal ordering used.
\end{prop}
\begin{Remark}
Our numeration of indices differs from the standard one where fermions have half-integer indices and bosons have integer indices.
This difference is nonessential for boson-fermion correspondence
\end{Remark}
Now we can write explicit formula for vectors $|(n+1/4)^2\rangle$ from Proposition \ref{prop:Fockdecomp}. 
\begin{prop}\label{prop:fermhww}
Highest weight vectors of $\Vir$ in $\mathrm{F}$  are given by
\begin{equation}\label{fermhww}
\left\{\begin{aligned}
|(n+1/4)^2\ra&=\prod_{k=1}^{n}\td{\psi}_{-k+1/4}\td{\psi}^*_{-k+1/4}|\varnothing,0\ra,\;\text{ for }  n\geq 0,\\ |(n+1/4)^2\ra&=\prod_{k=1}^{-n}\td{\psi}_{-k+3/4}\td{\psi}^*_{-k+3/4}|\varnothing,0\ra,\; \text{ for } n<0,                                        
\end{aligned}\right.
\end{equation} 

\end{prop}
\begin{proof}
First we prove that vectors defined in \eqref{fermhww} belong to $\mathrm{F}\subset \overline{\mathrm{F}}(0)$, i.e. depend only on $a_r$, $r \in \frac12\mathbb{Z}$. Indeed $a_j|(n+1/4)^2\ra=0$, $j \in \mathbb{Z}_{> 0}$ due to commutation relations (which follow from \eqref{bfc})
\begin{equation}\label{eq:[a,psi]}
[a_j,\td{\psi}_r]=\frac1{\sqrt{2}}\tilde{\psi}_{j+r},\qquad [a_j,\td{\psi}^*_r]=\frac{-1}{\sqrt{2}}\tilde{\psi}^*_{j+r}, \quad \text{ where } j\in \frac12\mathbb{Z},r\in\frac12\mathbb{Z}+\frac14.
\end{equation}
Also $a_0|(n+1/4)^2\ra=0$. Then we prove that these vectors satisfy \eqref{eq:n+1/4}. This can be done by direct calculation using formulas \eqref{La} and \eqref{eq:[a,psi]}. We will use another approach, introduce full Virasoro algebra generators by formula
\begin{equation}
L_n^{\text{full}}=\frac12\sum\limits_{j\in\frac12\mathbb{Z}} :a_{n-j}a_j:+\frac{1}{16}\delta_{n,0}. 
\end{equation}
Then we have
\begin{equation}
[L_n^{\text{full}},\td{\psi}_r]=(-\frac{n}2-r)\td{\psi}_{n+r},\quad [L_n^\text{full},\td{\psi}^*_r]=(-\frac{n}2-r)\td{\psi}^*_{n+r},
\end{equation}
Hence we get relation \eqref{eq:n+1/4} for $L_k^{\text{full}}$. But vectors $|(n+1/4)^2\ra$ do not depend on $a_j$, $j \in \mathbb{Z}$ so $L_k^{\text{full}}|(n+1/4)^2\ra=L_k|(n+1/4)^2\ra$, for $k\geq 0$.
\end{proof}

Natural scalar product on the space $\oplus_{k\in\mathbb{Z}}\overline{\mathrm{F}}(k)$ is defined by conjugation $\td{\psi}^+_r=\td{\psi}^*_{-r}$ and unit norm of $|\varnothing,0\ra$. Then it follows from \eqref{fermhww} that vectors $|(n+1/4)^2\ra$ have unit norm. Due to formulae \eqref{eq:a:psi} this product is consistent with Shapovalov form on the $\overline{\mathrm{F}}(k)$. Therefore this product coincides with Shapovalov form on Virasoro submodules from \eqref{Fockdecomp}. 

Now we return to the proof of \eqref{blockquarter}. Introduce vectors $|w(z)\ra^{\pm}=z^{1/16}e^{\pm 2\sqrt{2z}a_{-1/2}}|\varnothing\ra$. One can take any $\epsilon,\epsilon' \in \{-,+\}$ and calculate scalar product
\begin{equation*}
{}^{\epsilon}\la w(1)|w(z)\ra^{\epsilon'}=z^{1/16}\sum_{i=0}^{\infty}\la\varnothing|a_{1/2}^i a_{-1/2}^i|\varnothing\ra\frac{(8\epsilon\epsilon'\sqrt{z})^i}{i!^2}=z^{1/16}e^{4\epsilon\epsilon'\sqrt{z}}.
\end{equation*}
The result coincides with right side of \eqref{blockquarter}.

On the other hand due to \eqref{Fockdecomp} 
 the vectors $|w(z)\ra^{\epsilon}$ could be decomposed into orthogonal summands belonging to  $\pi^{(n+1/4)^2}_{\Vir},\, n\in\mathbb{Z}$.  Moreover
\begin{prop}
 	Vectors $|w(z)\ra^{\epsilon},\, \epsilon=\pm$ decompose into sum of the $\Vir$ Whittaker vectors, namely
 	\begin{equation}
 	|w(z)\ra^{\epsilon}=\sum_{n\in\mathbb{Z}}\mathfrak{l}^{\epsilon}_n |W_n(z)\ra,\quad |W_n(z)\ra\in \pi_{\Vir}^{(n+1/4)^2} \label{quartwirdec}
 	\end{equation}
 	with certain coefficients $\mathfrak{l}^{\pm}_n$ which do not depend on $z$.
 \end{prop}
 \begin{proof}
 	It is enough to check that 
 	\begin{equation*}
 	L_1|w(z)\ra^{\pm}=z|w(z)\ra^{\pm}, \quad L_2|w(z)\ra^{\pm}=0. 
 	\end{equation*}
 	This is done by standard calculation.
 	%For this calculation we use well-known formula
 	%\begin{equation}
 	%e^A B e^{-A}=e^{\mathrm{Ad} A}B 
 	%\end{equation}
 	%Calculation give us
 	%\begin{eqnarray}
% 	L_1|w(z)\ra^{\pm}=L_1 e^{2\sqrt{z}a_{-1/2}}|\varnothing\ra=4ze^{2\sqrt{z}a_{-1/2}}[a_{-1/2},[a_{-1/2},L_1]]|\varnothing\ra=z|w(z)\ra^{\pm}\\
% 	L_2|w(z)\ra^{\pm}=0, \quad \textrm{because} \quad [a_{-1/2},[a_{-1/2},L_2]]=0.
% 	\end{eqnarray}
 \end{proof}
It is clear from the definition of $|w(z)\ra^{\pm}$ that $
\mathfrak{l}_n^{-}=(-1)^n\mathfrak{l_n}^{+}.$ It follows from the definition that $\mathfrak{l}_n^{\epsilon}$ are given by formula 
\begin{equation}
\mathfrak{l}_n^{\epsilon}= \la (n+1/4)^2 |w(1)\ra^{\epsilon},\label{dirformula}
\end{equation}
and we calculate them using formulas \eqref{fermhww}. 

\begin{prop}\label{prop:l_nq}
	Coefficients $l_n^{+}$, $n\in\mathbb{Z}$ are given by formula $l_n^+=\sqrt{B_n}$.
%	\begin{equation}\label{eq:ln:algebraic}
%	l_n^+=\frac{2^{k(k+1)/2}}
%	{\prod\limits_{i=0}^{k-1}(2i+1)^{k-i}}  
%	\end{equation}
\end{prop}

\begin{proof}
Recall that under boson-fermion correspondence decomposable fermionic vectors correspond to Schur polynomials (\cite[Lec. 6.]{KacRaina}). In particular, vectors in right side of \eqref{fermhww} corresponds to polynomials with staircase diagram after substitution $p_j \mapsto \sqrt{2}a_{-j/2}$, where $p_j$ are power sum polynomials
\begin{equation}
|(n+1/4)^2\ra= S_{(k,k-1,\ldots 1)}\left(\sqrt{2}a_{-1/2},\sqrt{2}a_{-1},\sqrt{2}a_{-3/2},\ldots\right)|\varnothing\ra, \quad k=\left\{\begin{aligned}
2n,\, n>0\\   -2n-1,\, n<0,
\end{aligned}\right.\label{boshww}
\end{equation}
Schur polynomials satisfy 
$\langle p_1^{|\lambda|},S_\lambda\rangle=|\lambda|!/h(\lambda) $, where $h(\lambda)=\prod_{s \in \lambda}h(s)$ product of hook length (see e.g. \cite[Sec. 1.4, Ex. 3]{MacdonaldBook}). Therefore 
\begin{equation*}
l_n^+=\la\varnothing| S_{(k,k-1,\ldots)}(\{\sqrt{2} a_{-r}\})(\frac{(2\sqrt{2})^N}{N!}a_{-1/2})^{N}|\varnothing\ra=2^{N}/
{\prod\nolimits_{i=0}^{k-1}(2i+1)^{k-i}}=\sqrt{B_n},
\end{equation*}
where $N=k(k+1)/2$.
\end{proof}
Now we calculate ${}^{\epsilon}\la w(1)|w(z)\ra^{\epsilon'}$
using decomposition \eqref{quartwirdec} and Proposition \ref{prop:l_nq} and get 
\begin{equation*}
{}^{\epsilon}\la w(1)|w(z)\ra^{\epsilon'}=\sum_{n\in\mathbb{Z}} (\epsilon\epsilon')^n B_n \mathcal{F}((n+1/4)^2|z).
\end{equation*}
%Coefficients $\mathfrak{l}_n^{\epsilon}$ satisfy evident relation $\mathfrak{l}_n^{-}=(-1)^n\mathfrak{l_n}^{+}$.
%So up to coefficients $\mathfrak{l}_n^{+}$ we have \eqref{blockquarter}
which coincides with the left side of \eqref{blockquarter}.

\begin{Remark}
	Actually the space $\mathrm{F}$ has a natural structure of the basic module over $\widehat{\mathfrak{sl}}(2)$. In particular one can introduce operator $h_0$ such that $h_0v=2nv$ if $v \in \pi_{\Vir}^{(n+1/4)^2}$. Then due to \eqref{Kiev} and calculations above one can write Painlev\'e III($D_8$) $\tau$ function for $\sigma=1/4$  as 
\begin{equation}
\tau(1/4,s|z)\propto \sum_{n\in\mathbb{Z}}B_n(-s)^n \kr{F}((n+1/4)^2|z)={}^-\langle w(1)|s^{h_0/2}|w(z) \rangle^+.
\end{equation}
Taking into account AGT correspondence \cite{AGT} the right side of this formula coincides with the dual partition function introduced in \cite{NO:2003} (see eq. (5.25) in loc. cit.).
\end{Remark}

\section{Bilinear relations from the algebra $\sfF\oplus\NSR$}
\label{sec:FxNSR}
\subsection{Algebra $\sfF\oplus\NSR$ and its conformal blocks}
\label{ssec:FxNSR}
The $\sfF\oplus \NSR$ algebra is a direct sum of the free-fermion algebra $\mathsf{F}$ with generators $f_r$ ($r\in\mathbb{Z}+\delta$) and $\NSR$ (Neveu-Schwarz-Ramond or Super Virasoro) algebra with generators $L_n, G_r$ ($n\in\mathbb{Z}, r\in\mathbb{Z}+\delta$, $\delta=0,1/2$). These generators satisfy commutation relations
\begin{equation}
\begin{aligned}
&\{f_r,f_s\}=\delta_{r+s,0}, \quad \{f_r,G_s\}=0 \\
&[L_n,L_m]=(n-m)L_{n+m}+ \frac{(n^3-n)}8c_{\nsr}\delta_{n+m,0}\\
&\{G_r,G_s\}=2L_{r+s}+\frac{1}{2}c_{\nsr}\left(r^2-\frac{1}{4}\right)\delta_{r+s,0}\\~
&[L_n,G_r]=\left(\frac{1}{2}n-r\right)G_{n+r}.
\end{aligned}\label{def:FxNSR}
\end{equation}
It is convenient to express the central charge by
\[
c_{\nsr}=1+2Q^2,\quad Q=b+b^{-1}.
\]
Case of $\delta=1/2$ i.e. half-integer indices $r$ of $G_r$ and $f_r$ is called $\BNS$ sector of defined above algebras.
Case of $\delta=0$ i.e. integer indices is called $\BR$ sector. Remark that $c_{\nsr}$ differs from the Virasoro central charge $c=\frac32c_{\nsr}$.

Below we will use parametrization of the highest weight $\Delta$ by
\begin{equation}\label{Delta}
\D^{\delta}=\frac{1-2\delta}{16}+\frac12\left(\frac{Q^2}{4}-P^2\right), 
\end{equation}
where $\D^{\delta}\equiv \D^{\NS}$ if $\delta=1/2$, and $\D^{\delta}\equiv \D^{\R}$ if $\delta=0$.

In $\BNS$ case we denote by $\pi^{\D^{\NS}}_{\sfF\oplus \NSR}$ a Verma module of the $\sfF\oplus \NSR$ algebra. This module is isomorphic to a tensor product of Verma modules $\pi^{\NS}_{\sfF}$ and $\pi^{\D^{\NS}}_\NSR$ which are generated by the highest weight vectors 
$|1\ra$ and $|\D^{\NS}\ra$ defined by
\[
f_r|1\ra=0, \quad r>0,
\]
and
\[
L_0|\D^{\NS}\ra=\D^{\NS}|\D^{\NS}\ra; \qquad L_n|\D^{\NS}\ra=0, \quad G_r|\D^{\NS}\ra=0, \quad n,r>0.
\]
%The representation space is spanned by vectors obtained by the action of generators with negative indices on the highest weight vector. 
We denote the highest weight vector  $|1\ra\otimes |\D^{\NS}\ra$ as $\overline{|\D^{\NS}\ra}$.

In the $\BR$ sector we have an analogous construction $\pi^{\D^{\R}}_{\sfF\oplus\NSR}=\pi_{\sfF}^{\R}\otimes\pi^{\D^{\R}}_{\sfF\oplus\NSR}$. 
Here $\pi_{\sfF}^{\R}$ is a Verma module with two highest weight vectors $|1^{\pm}\ra$, defined by
\begin{equation}
f_r|1^{\pm}\ra=0, \quad r>0,\qquad
f_0|1^{\pm}\ra=\frac{1}{\sqrt2}|1^{\mp}\ra.\label{hwwF}
\end{equation}
By $\pi^{\D^{\R}}_{\NSR}$ we denote Verma module with highest weight vectors $|\D^{\R},\pm\ra$ defined by
\begin{equation}
\begin{aligned}
G_r|\D^{\R},\pm\ra=0, \quad r>0, \qquad
L_n|\D^{\R},\pm\ra=0, \quad n>0,\\
G_0|\D^{\R},\pm\ra=-\frac{iP}{\sqrt2}|\D^{\R},\mp\ra, \qquad
L_0|\D^{\R},\pm\ra=\D^{\R}|\D^{\R},\pm\ra.\label{hwwR}
\end{aligned}
\end{equation}
Actually the formula for action $L_0$ follows from relation $G_0^2=L_0-c_{\nsr}/16$ and parametrization \eqref{Delta}.

Let us denote highest weight vectors of $\pi^{\D^{\R}}_{\sfF\oplus\NSR}$ as
\begin{equation}
|1^{\mu}\ra\otimes|\D^{\R},\nu\ra=\overline{|\mu;\D^{\R},\nu\ra},\quad \mu,\nu=\pm. 
\end{equation}
We define action of $G_0$ on the highest weight vector $\overline{|\mu;\D^{\R},\nu\ra}$ by
\begin{equation*}
G_0\overline{|+;\D^{\R},\nu\ra}=|1^+\ra\otimes G_0|\D^{\R},\nu\ra, \quad, \nu=\pm. \label{anticomm}
\end{equation*}
Then due to anticommutativity of $f_0$ and $G_0$
\begin{equation}
G_0\overline{|\mu;\D^{\R},\nu\ra}=\mu |1^\mu\ra\otimes G_0|\D^{\R},\nu\ra, \quad \mu, \nu=\pm \label{eq:hww_conj} . 
\end{equation}

Now let us define irregular limit of conformal block for $\NSR$ algebra.
At first we define $\NSR$ Whittaker vectors $|W_{\delta}(z)\ra$ for both sectors
(see \cite{Ito} for $\BR$ sector, \cite{Belavin:2007zz} for $\BNS$ sector)
\begin{eqnarray}
|W_{\NS}(z)\ra=z^{\D^{\NS}}\sum\limits_{2N=0}^{\infty}z^{N}|N\ra,\qquad |N\ra^{\NS}\in\pi^{\D^{\NS}}_{\NSR} ,\qquad L_0|N\ra^{\NS}=(\D^{\NS}+N)|N\ra^{\NS},\\
|W_{\R,\pm}(z)\ra=z^{\D^{\R}}\sum\limits_{2N=0}^{\infty}z^{N}|N\ra^{\R,\pm},\qquad |N\ra^{\R,\pm}\in\pi^{\D^{\R}}_{\NSR} ,\qquad L_0|N\ra^{\R,\pm}=(\D^{\R}+N)|N\ra^{\R,\pm},
\end{eqnarray}
where vectors $|N\ra^{\NS}$ and $|N\ra^{\R,\pm}$ satisfy
\begin{equation}
G_{1/2}|N\ra^{\NS}=|N-1/2\ra^{\NS}, \quad N>0, \quad 
G_{3/2}|N\ra^{\NS}=0,
\end{equation} and
\begin{equation}
L_{1}|N\ra^{\R,\pm}=1/2|N-1\ra^{\R,\pm}, \quad N>0, \quad 
G_{1}|N\ra^{\R,\pm}=0.\label{WhitaxR}
\end{equation} 
Equivalently, in terms of Whittaker vectors last equations could be written as
\begin{equation}
G_{1/2}|W_{\NS}(z)\ra=z^{1/2}|W_{\NS}(z)\ra, \quad G_{r}|W_{\NS}(z)\ra=0, r\geq3/2
\end{equation}  for $\BNS$ sector and
\begin{equation}
L_{1}|W_{\R,\pm}(z)\ra=\frac12z|W_{\R,\pm}(z)\ra, \quad G_r|W_{\R,\pm}(z)\ra=0, r>0 \label{WhitaxRW}
\end{equation} for $\BR$ sector.
These conditions coupled to normalization of $|0\ra^{\NS}=|\D^{\NS}\ra,\quad |0\ra^{\R,\pm}=|\D^{\R},\pm\ra$, 
determine Whittaker vectors completely for generic values of $\Delta^\delta$ as in Virasoro case. To be more precise we say that $P$ is generic if $P\not \in \{\frac12(mb+nb^{-1})|m,n \in \mathbb{Z}\}$.
For generic $P$ Shapovalov form (see below) is nondegenerate and the Verma module is irreducible \cite{IKS}. And for corresponding values of $\Delta$ the Whitteker vector is well defined.  

Formulas for action of the operators $L_k, G_r$ for $k>1, r>1/2$ on $|W_{\BNS}(z)\ra, |W_{\R,\pm}(z)\ra$ follow from the $\NSR$ commutation relations.

We define complex Shapovalov form on $\pi^{\D^{\NS}}_{\sfF\oplus\NSR}$ and $\pi^{\D^{\R}}_{\sfF\oplus\NSR}$
in which conjugation of operators is
\begin{equation}
L_{n}^+=L_{-n},\quad
G_{r}^+=G_{-r},\quad
f_{r}^+=-f_{-r} \quad\label{eq:conj:operators}
\end{equation}
In order to define scalar product completely we should fix Shapovalov form on highest weight vectors.
In $\BNS$ sector we have only one highest weight vector $\overline{|\D^{\NS}\ra}$ which we normalize so that
$\overline{\la\D^{\NS}}|\overline{\D^{\NS}\ra}=1$
In $\BR$ sector we have additional scalar products between different highest weight vectors. 
Due to skew-symmetry of $f_0$ operator we have that $\la 1^+|1^-\ra=0$ and $\la 1^+| 1^+\ra=-\la 1^-| 1^-\ra$. 
We normalize $\la 1^+| 1^+\ra=-\la 1^-| 1^-\ra=1$. Also we state $\la \D^{\R},+|\D^{\R},-\ra=0$ and
$\la \D^{\R},\pm|\D^{\R},\pm\ra=1$.

The irregular (or Whittaker, or Gaiotto) limit of $\NSR$ conformal block is defined by
\begin{equation}\label{eq:NSRWhitcom}
\kr{F}_{c_{\nsr}}(\D^{\NS}|z)=\la W_{\NS}(1)|W_{\NS}(z)\ra,\quad \kr{F}_{c_{\nsr}}(\D^{\R}|z)=\la W_{\R,\pm}(1)|W_{\R,\pm}(z)\ra,
\end{equation}
where the $\BR$ conformal block does not depend on subscript of the $\BR$ Whittaker vector. That is because properties \eqref{hwwR}
have symmetry in interchange of vectors $|\D^{\R},\pm\ra$.

\begin{Remark}
Scalar product $\la \D^{\R},+|\D^{\R},-\ra$ is not zero in general, but one can choose different pair of highest weight vectors
using the formulas
\begin{equation*}
|\D^{\R},+\ra'=|\D^{\R},+\ra-\alpha |\D^{\R},-\ra, \qquad
|\D^{\R},-\ra'=|\D^{\R},-\ra-\alpha |\D^{\R},+\ra.      
\end{equation*}
where non-primed vectors are normalized on $1$.
Indeed $|\D^{\R},\pm\ra'$ satisfies \eqref{hwwR} just as well as $|\D^{\R},\pm\ra$.
Moreover, if $\alpha$ is a solution of the equation
\begin{equation*}
0=(1+\alpha^2)\la \D^{\R},+|\D^{\R},-\ra-2\alpha
\end{equation*}
then we get ${}'\la\D^{\R},-|\D^{\R},+\ra'=0$.
\end{Remark}

\subsection{Verma module decomposition}
\label{ssec:decomp}
Let us recall the free-field realization of the $\NSR$ algebra. Consider the algebra generated by $c_n, n\in\mathbb{Z}$ and $\psi_r, r\in\mathbb{Z}+\delta$ (free boson and free fermion) with relations
\[
[c_n,c_m]=n\delta_{n+m,0}, \quad
[c_n,\psi_r]=0, \quad
\{\psi_r,\psi_s\}=\delta_{r+s,0}.
\]
Consider two sets of such generators which we will distinguish by superscript $\mp$ and
add to them zero mode $c_0^{\mp}$ as $\mp i\hat{P}$. We will omit the superscript when it is not confusing.
Then a Fock representation of this algebra is generated by a vacuum vector $|P\ra$
for $\BNS$ sector and vacuum vectors $|P\ra^{\pm}$ for $\BR$ sector. 
Vector $|P\ra$ satisfy $\psi_r|P\ra=c_n|P\ra=0$, $\hat{P}|P\ra=P|P\ra$ for $r,n >0$ and analogously for vectors $|P\ra^{\pm}$ .
For $\BR$ case we should also add $\psi_0^{\mp}|P\ra^{\mu}=\pm (1/\sqrt{2})|P\ra^{-\mu}$, $\mu =\pm$.
On this Fock module we can define an action of the $\NSR$ algebra by the formulae
\begin{equation}
\begin{aligned}
&L_n=\frac12\sum_{k\neq 0,n}c_k c_{n-k}+\frac12\sum_{r} r\psi_{n-r}\psi_r+\frac{i}{2}\left(Q n\mp 2\hat{P}\right)c_n,\quad n\neq 0,\\
&L_0=\sum_{k>0} c_{-k} c_k+\sum_{r>0}r \psi_{-r}\psi_r+\frac{1-2\delta}{16}+\frac{1}{2} \left(\frac{Q^2}{4}-\hat{P}^2 \right),\label{eq.FFR} \\
&G_r=\sum_{n\neq 0}c_n \psi_{r-n}+i (Q r\mp \hat{P})\psi_r.
\end{aligned}
\end{equation}
Operators $c_n^{\mp}, \psi_r^{\mp}$ with two different signs give us two different free-field realizations of $\NSR$ algebra.

Recall that we say that $P$ is generic if $P\not \in \{\frac12(mb+nb^{-1})|m,n \in \mathbb{Z}\}$.
For generic $P$ the $\NSR$ module defined by \eqref{eq.FFR} is irreducible (see e.g. Case V in \cite{IKS}) and is isomorphic to the Verma module $\pi_{\NSR}^{\Delta^{\NS}}$
in $\BNS$ sector and to $\pi^{\D^{\R}}_{\NSR}$ in $\BR$ sector. Weights $\D^{\NS}$ and $\D^{\R}$ are defined by \eqref{Delta}.
Operators $c_n^{\mp}, \psi_r^{\mp}$ acting on
$\pi^{\D^{\delta}}_{\NSR}$ are conjugated by the so-called
super Liouville reflection operator.

%The sign $\mp$ in \eqref{eq.FFR} refers to the existence of two free-field representations of the same Verma module. 
%We denote the corresponding generators by $c_n^-, \psi_r^-$ and $c_n^+, \psi_r^+$. 

As the main tool we will use action of the algebra $\Vir\oplus \Vir$ on the representations $\pi^{\D^{\delta}}_{\sfF\oplus\NSR}$, $\delta=0,1/2$
(following \cite{CSS}, \cite{Lashkevich}).
The generators of the algebra $\Vir\oplus \Vir$ are defined by formulas
\begin{equation}
\begin{aligned}
&L_{n}^{\scriptscriptstyle{(1)}}=\frac{b^{-1}}{b^{-1}-b}L_{n}-\frac{b^{-1}+2b}{b^{-1}-b}\left(\frac12\sum\limits_{r\in\mathbb{Z}+\delta}r:f_{n-r}f_{r}:+\frac{1-2\delta}{16}\delta_{n,0}\right)+\frac{1}{b^{-1}-b}\sum\limits_{r\in\mathbb{Z}+\delta}f_{n-r}G_r,\\
&L_{n}^{\scriptscriptstyle{(2)}}=\frac{b}{b-b^{-1}}L_{n}-\frac{b+2b^{-1}}{b-b^{-1}}\left( \frac12\sum\limits_{r\in\mathbb{Z}+\delta}r:f_{n-r}f_{r}:+\frac{1-2\delta}{16}\delta_{n,0}\right)+\frac{1}{b-b^{-1}}\sum\limits_{r\in\mathbb{Z}+\delta}f_{n-r}G_r \label{Vir12}.
\end{aligned}
\end{equation}
The expression in parentheses from second summand is just expression for $L_n^f$ which define representation of $\Vir$ $c=1/2$ algebra on $\pi_{\sfF}$
in appropriate sector.
Note that the expressions for $L_{n}^{(\eta)}$, $\eta=1,2$ contain infinite sums and belong to certain completion of the universal enveloping algebra
of $\sfF\oplus \NSR$. %The operators $L_{n}^{(\eta)}$ act on any highest weight representation of $\sfF\oplus \NSR$. 
Conjugation of $L_{n}^{(\eta)}$ is standard:
$L_{n}^{(\eta)+}=L_{-n}^{(\eta)}$, due to \eqref{eq:conj:operators}.

It is convenient to express the central charge and the highest weights of the Virasoro algebra by
\begin{align}
\Delta(P,b)=\frac{Q^2}{4}-P^2, \;\; c(b)=1+6Q^2, \quad \text{where} \quad Q=b+b^{-1}. \label{eq:cDelta:Vir}
\end{align} 
Then the central charges of these $\Vir^{\scriptscriptstyle{(1)}}$ and $\Vir^{\scriptscriptstyle{(2)}}$ subalgebras are equal to 
\begin{align}
c^{(\eta)}=c(b^{(\eta)}),\; \eta=1,2,\quad \text{where}\quad  (b^{\scriptscriptstyle{(1)}})^2=\frac{2 b^2}{1-b^2}, \quad (b^{\scriptscriptstyle{(2)}})^{-2}=\frac{2 b^{-2}}{1-b^{-2}}. \label{cc}
\end{align}
Note that the symmetry $b\leftrightarrow b^{-1}$ permutes $\Vir^{\scriptscriptstyle{(1)}}$ and $\Vir^{\scriptscriptstyle{(2)}}$. Here and below $b^2\neq 0,1$.

Now consider the space $\pi^{\D^{\NS}}_{\sfF\oplus \NSR}$ as a representation of $\Vir \oplus \Vir$. Clearly, the vector $\overline{|\Delta^{\NS}\ra}=|1\ra\otimes |\Delta^{\NS}\ra$ is a highest weight vector with respect to $\Vir \oplus \Vir$.
This vector generates a Verma module $\pi^{\D^{\scriptscriptstyle{(1)}},\D^{\scriptscriptstyle{(2)}}}_{\Vir\oplus \Vir}$. 
The highest weight $(\D^{\scriptscriptstyle{(1)}},\D^{\scriptscriptstyle{(2)}})$ can be found from \eqref{Vir12}, namely
\begin{align}\label{eq:DeltaVirNSR}
\D^{\scriptscriptstyle{(1)}}=\frac{b^{-1}}{b^{-1}-b}\D^{\NS}, \quad  \D^{\scriptscriptstyle{(2)}}=\frac{b}{b-b^{-1}}\D^{\NS}
\end{align}
The whole space $\pi^{\D^{\NS}}_{\sfF\oplus \NSR}$ is larger than $\pi^{\D^{\scriptscriptstyle{(1)}},\D^{\scriptscriptstyle{(2)}}}_{\Vir\oplus \Vir}$. The following decomposition was proven in \cite{BBFLT}. 

\begin{thm} \label{pr:FNSR}
 For generic $P$ the space $\pi^{\D^{\NS}}_{\sfF\oplus \NSR}$ is isomorphic to the sum of $\Vir\oplus \Vir$ modules
\begin{equation}
\begin{aligned}
\pi^{\D^{\NS}}_{\sfF\oplus \NSR}\cong\bigoplus\limits_{2n\in\mathbb{Z}}\pi^n_{\Vir\oplus \Vir} \label{FxNSR}.
\end{aligned}
\end{equation}
The highest weight $(\D^{\scriptscriptstyle{(1)}}_{n},\D^{\scriptscriptstyle{(2)}}_{n})$ of the Verma module $\pi^n_{\Vir\oplus \Vir}$ is defined by $\Delta^{(\eta)}_{n}=\Delta(P^{(\eta)}_n,b^{(\eta)}),$ $\eta=1,2$, 
where 
\begin{align}
P^{\scriptscriptstyle{(1)}}_{n}=P^{\scriptscriptstyle{(1)}}+nb^{\scriptscriptstyle{(1)}}, \; P^{\scriptscriptstyle{(2)}}_{n}=P^{\scriptscriptstyle{(2)}}+n\left(b^{\scriptscriptstyle{(2)}}\right)^{-1}, \qquad P^{\scriptscriptstyle{(1)}}=\frac{P}{\sqrt{2-2 b^2}}, \; P^{\scriptscriptstyle{(2)}}=\frac{P}{\sqrt{2-2 b^{-2}}}\label{momentum}.
\end{align}
\end{thm}

\begin{figure}[h]
\begin{center}
\begin{tikzpicture}[x=4em,y=4em]
\draw[dashed] (-3,0.5) -- (3,0.5);
\draw (-2,0.45) -- (-2,0.5) node[anchor=south] {$-1$};
\draw (-1,0.45) -- (-1,0.5) node[anchor=south] {$-0.5$};
\draw (0,0.45) -- (-0,0.5) node[anchor=south] {$0$};
\draw (1,0.45) -- (1,0.5) node[anchor=south] {$0.5$};
\draw (2,0.45) -- (2,0.55) node[anchor=south] {$1$};
\draw (3.2,0.5) node[anchor=west] {$n$};
\draw[dashed] (-3.3,-3.2) -- (-3.3,0.2);
\draw (-3.3,-3.3) node[anchor=north] {$L_0+L_0^f$};
\draw (-3.25,0) -- (-3.35,0) node[anchor=east] {$\D^{\NS}$};
\draw (-3.25,-0.5) -- (-3.35,-0.5) node[anchor=east] {$\D^{\NS}+0.5$};
\draw (-3.25,-1) -- (-3.35,-1) node[anchor=east] {$\D^{\NS}+1$};
\draw (-3.25,-1.5) -- (-3.35,-1.5) node[anchor=east] {$\D^{\NS}+1.5$};
\draw (-3.25,-2) -- (-3.35,-2) node[anchor=east] {$\D^{\NS}+2$};
\draw (-3.25,-2.5) -- (-3.35,-2.5) node[anchor=east] {$\D^{\NS}+2.5$};
\draw (1.5,-3.5) -- (2,-2) -- (2.5,-3.5);
\draw (1.4,-2.1) -- (1,-0.5) -- (0.6,-2.1);
\draw (-0.7,-2.7) -- (0,0) -- (0.7,-2.7);
\draw (-1.4,-2.1) -- (-1,-0.5) -- (-0.6,-2.1);
\draw (-1.5,-3.5) -- (-2,-2) -- (-2.5,-3.5);
\end{tikzpicture}
\caption{\leftskip=1in \rightskip=1in Decomposition of $\pi^{\D^{\NS}}_{\sfF\oplus \NSR}$ into direct sum of representations of the algebra $\mathsf{Vir}\oplus\mathsf{Vir}$. Each interior angle corresponds to the Verma module $\pi^n_{\Vir\oplus \Vir}$.}
	\label{NS-decompos-pic}
\end{center}
\end{figure}

Highest weight vectors of $\pi^{n}_{\Vir\oplus\Vir}$ are given by the formula
\begin{equation}
|P,n\ra\propto \prod_{r=1/2}^{(4n-1)/2}\chi_r^-\overline{|\D^{\NS}\ra},\, n>0, \quad |P,n\ra\propto \prod_{r=1/2}^{(-4n-1)/2}\chi_r^+\overline{|\D^{\NS}\ra},\, n<0,
\quad |P,0\ra=\overline{|\D^{\NS}\ra},\label{hwwdec}
\end{equation}
where $\chi_r^{\mp}=f_r-i\psi_r^{\mp}$.

Let us now formulate and prove analogous decomposition for the Ramond sector.
Introduce parity operator $\mathfrak{P}:\pi^{\D^{\R}}_{\sfF\oplus\NSR}\mapsto \pi^{\D^{\R}}_{\sfF\oplus\NSR}$.
Operator $\mathfrak{P}$ is defined by the action $\mathfrak{P}\overline{|+;\D^{\R},+\ra}=\overline{|+;\D^{\R},+\ra}$ and property that $\mathfrak{P}$
commutes with even operators and anticommutes with odd operators of $\sfF\oplus\NSR$ algebra.
%Addition of such $\mathfrak{P}$ is consistent with Lie superalgebra axioms.

Evidently $\mathfrak{P}^2=1$ and the operators
$\frac{1\pm \mathfrak{P}}2$ are projectors on $\mathfrak{P}$ eigenspaces with eigenvalues $1$ and $-1$ which we will call even and odd subspaces respectively.
We will mark objects related to these subspaces by index $\epsilon$ equals $0$ for the even subspace and $1$ for the odd subspace. 
So $\pi^{\D^{\R}}_{\sfF\oplus\NSR}$ decompose into direct sum of $\pi^{\D^{\R},0}_{\sfF\oplus\NSR}$ and $\pi^{\D^{\R},1}_{\sfF\oplus\NSR}$.
 
Operator $\mathfrak{P}$ also could be defined on $\pi_{\sfF}^{\R}$ and $\pi^{\D^{\R}}_{\NSR}$ separately in an obvious way. Above defined
operator is just the tensor product of these operators. To make last definition precise we need to add condition $\mathfrak{P}|1^+\ra=|1^+\ra$.

It appears that summands $\pi^{\D^{\R}, \epsilon}_{\sfF\oplus\NSR}$ decompose into direct sum of $\Vir\oplus\Vir$ modules.
Namely we have
\begin{thm} \label{pr:FNSRR}
 For generic $P$ the space $\pi^{\D^{\R}}_{\sfF\oplus \NSR}$ is isomorphic to the sum of $\Vir\oplus \Vir$ module
\begin{equation}
\begin{aligned}
\pi^{\D^{\R}}_{\sfF \oplus \NSR}\cong \pi^{\D^{\R},0}_{\sfF \oplus \NSR}\oplus \pi^{\D^{\R},1}_{\sfF \oplus \NSR}\cong\bigoplus\limits_{2n+1/2\in\mathbb{Z}}\pi^{n,0}_{\Vir\oplus \Vir}\oplus \bigoplus\limits_{2n+1/2\in\mathbb{Z}}\pi^{n,1}_{\Vir\oplus \Vir} \label{FxNSRR}.
\end{aligned}
\end{equation}
where highest weights of modules $\pi^{n,\epsilon}_{\Vir\oplus \Vir}$ are $(\D^{\scriptscriptstyle{(1)}}_{n},\D^{\scriptscriptstyle{(2)}}_{n})$,
defined in Theorem \ref{pr:FNSR}. Superscript $\epsilon=0,1$ denotes parity. 
\end{thm}

\begin{figure}[h]
\begin{center}
\begin{tikzpicture}[x=4em,y=4em]
\draw[dashed] (-3.5,0.5) -- (3.5,0.5);
\draw (-2.5,0.45) -- (-2.5,0.55) node[anchor=south] {$-1.25$};
\draw (-1.5,0.45) -- (-1.5,0.55) node[anchor=south] {$-0.75$};
\draw (-0.5,0.45) -- (-0.5,0.55) node[anchor=south] {$-0.25$};
\draw (0.5,0.45) -- (0.5,0.55) node[anchor=south] {$0.25$};
\draw (1.5,0.45) -- (1.5,0.55) node[anchor=south] {$0.75$};
\draw (2.5,0.45) -- (2.5,0.55) node[anchor=south] {$1.25$};
\draw (3.7,0.5) node[anchor=west] {$n$};
\draw[dashed] (-3.3,-3.7) -- (-3.3,0.2);
\draw (-3.3,-4.0) node[anchor=north] {$L_0+L_0^f-1/16$};
\draw (-3.25,0) -- (-3.35,0) node[anchor=east] {$\D^{\R}$};
\draw (-3.25,-0.5) -- (-3.35,-0.5) node[anchor=east] {$\D^{\R}+0.5$};
\draw (-3.25,-1) -- (-3.35,-1) node[anchor=east] {$\D^{\R}+1$};
\draw (-3.25,-1.5) -- (-3.35,-1.5) node[anchor=east] {$\D^{\R}+1.5$};
\draw (-3.25,-2) -- (-3.35,-2) node[anchor=east] {$\D^{\R}+2$};
\draw (-3.25,-2.5) -- (-3.35,-2.5) node[anchor=east] {$\D^{\R}+2.5$};
\draw (-3.25,-3) -- (-3.35,-3) node[anchor=east] {$\D^{\R}+3.0$};
\draw (-3.25,-3.5) -- (-3.35,-3.5) node[anchor=east] {$\D^{\R}+3.5$};
\draw (2.0,-4.0) -- (2.5,-3) -- (3.0,-4.0);
\draw (1.9,-2.1) -- (1.5,-1.0) -- (1.1,-2.1);
\draw (0.05,-1.7) -- (0.5,0) -- (0.95,-1.7);
\draw (-0.95,-1.7) -- (-0.5,0) -- (-0.05,-1.7);
\draw (-1.9,-2.1) -- (-1.5,-1.0) -- (-1.1,-2.1);
\draw (-2.0,-4.0) -- (-2.5,-3) -- (-3.0,-4.0);
\end{tikzpicture}
\caption{\leftskip=1in \rightskip=1in Decomposition of $\pi^{\D^{\R},\epsilon}_{\sfF\oplus \NSR}$ into direct sum of representations of the algebra $\mathsf{Vir}\oplus\mathsf{Vir}$ (for each $\epsilon=0,1$). Each interior angle corresponds to the Verma module $\pi^{n,\epsilon}_{\Vir\oplus \Vir}$.}
	\label{R-decompos-pic}
\end{center}
\end{figure}
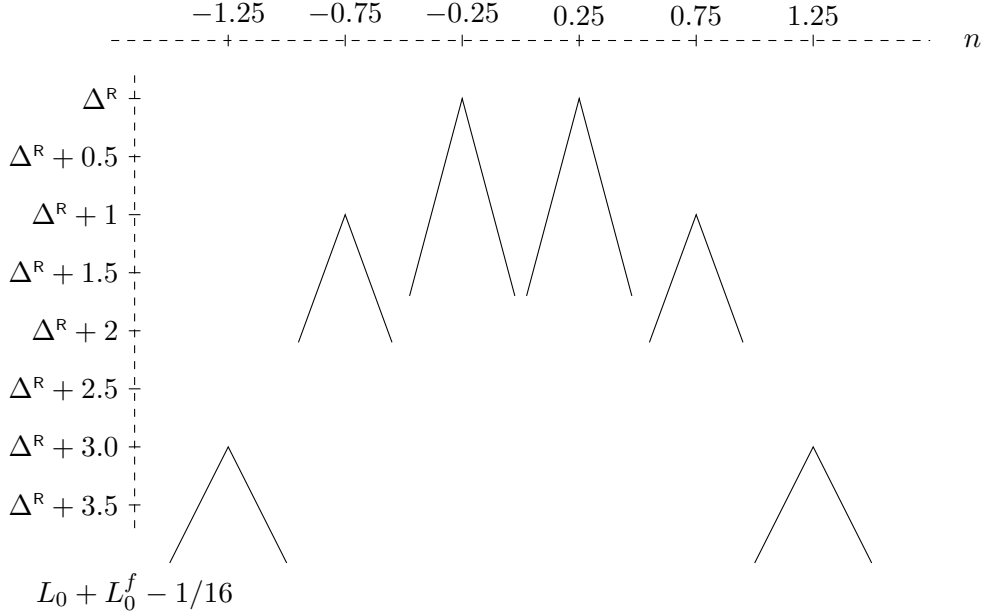

\begin{proof}
Arguments are analogous to the proof of Theorem \ref{pr:FNSR} (see \cite{BBFLT}, \cite{KGP}). 

At first let us find highest weight vectors for $\Vir\oplus\Vir$ algebra in $\pi^{\D^{\R}}_{\sfF\oplus\NSR}$.
We will denote the highest weight vectors of $\pi^{n,\epsilon}_{\Vir\oplus\Vir}, \epsilon=0,1$
with weights $\Delta^{(\eta)}_{n}, \eta=1,2$ as $|P,n\ra_{\epsilon}$.
 
In $\pi^{\D^{\R}}_{\sfF\oplus\NSR}$ we have four highest weight
vectors $\overline{|\mu;\D^{\R},\nu\ra}$, $\mu, \nu=\pm$.
Clearly $L_k^{(\eta)}\overline{|\mu;\D^{\R},\nu\ra}=0$, $\eta=1,2,\, k>0$.
Vectors $\overline{|\mu;\D^{\R},\mu\ra}$, $\mu=\pm$ are even and vectors
$\overline{|\mu;\D^{\R},-\mu\ra}$ are odd.
The following linear combinations are eigenvectors for $L_0^{(\eta)}$, $\eta=1,2$, i.e. highest weight vectors of $\Vir\oplus\Vir$ algebra
\begin{equation}
\begin{aligned}
&|P,\pm 1/4\ra_1\propto\frac{1}{\sqrt2}(i\,\overline{|-;\D^{\R},+\ra}\pm \overline{|+;\D^{\R},-\ra})= i\chi_0^{\mp}\overline{|+;\D^{\R},+\ra}\\
&L_0^{(\eta)}|P,\pm 1/4\ra_1=\D^{(\eta)}_{\pm 1/4}|P,\pm 1/4\ra_1,\\
&|P,\pm 1/4\ra_0\propto\frac{1}{\sqrt2}(\overline{|+;\D^{\R},+\ra}\pm i \, \overline{|-;\D^{\R},-\ra})=\frac{1}{\sqrt2} \chi_0^{\mp}\chi_0^{\pm}
\overline{|+;\D^{\R},+\ra}\\
&L_0^{(\eta)}|P,\pm 1/4\ra_0=\D^{(\eta)}_{\pm 1/4}|P,\pm 1/4\ra_0,\label{starthwwR}
\end{aligned}
\end{equation}
where we used $\chi_r^{\mp}=f_r-i\psi_r^{\mp}$ as for $\BNS$ case but with integer indices.
Using \eqref{Vir12} and \eqref{eq.FFR} one can calculate commutation relations between $L_n^{(\eta)}$, $\eta=1,2$ and $\chi_r$
\begin{equation}
\begin{aligned}
{}&[L_n^{(1)}+L_n^{(2)},\chi_r]=-\left(\frac{n}2+r\right)\chi_{r+n}, \\ &{} [b L_n^{(1)}+b^{-1} L_n^{(2)},\chi_r]=-((n+r)Q\mp\hat{P})\chi_{r+n}+i\sum_{m\neq 0}c_m\chi_{r+n-m}.\label{chicomm}
\end{aligned}
\end{equation}
(cf. \cite[Eq. (3.23)]{BBFLT} for $\BNS$ sector).
Then it is easy to check that for $2n\in\mathbb{Z}, n\geq0$ vectors
\begin{equation}
\begin{aligned}
&|P,\pm(1/4+n)\ra_{2n\bmod 2} \propto\left(\prod\limits_{r=0}^{2n}\chi_{-r}^{\mp}\right)|1^+\ra\otimes|\D^{\R},+\ra\\
&|P,\pm(1/4+n)\ra_{(2n+1)\bmod 2} \propto\left(\prod\limits_{r=0}^{2n}\chi_{-r}^{\mp}\right)\chi_0^{\pm}|1^+\ra\otimes|\D^{\R},+\ra\label{hwwdecR}
\end{aligned}
\end{equation}
are the highest weight vectors for $\Vir\oplus\Vir$ algebra with appropriate highest weight. Given by \eqref{Vir12} generators of $\Vir\oplus\Vir$
algebra are even, so vectors $|P,n\ra_{\epsilon}$, $2n+1/2\in\mathbb{Z}$ generate Verma modules $\pi^{n,\epsilon}_{\Vir\oplus\Vir}$ with parity
$\epsilon$. For generic $P$ Verma modules $\pi^{n,\epsilon}_{\Vir\oplus\Vir}$ with certain parity $\epsilon$ are irreducible and
linear independent. Linear independence between $0$ and $1$ modules with the same highest weight follows from different parity of modules.

From \eqref{momentum} follows that
\begin{equation*}
\D^{(1)}_{n}+\D^{(2)}_{n}=\D^{\R}-1/16+2n^2. 
\end{equation*}

In order to finish the prove of decomposition \eqref{FxNSRR} it is sufficient to check equality of characters. Indeed
\begin{equation*}
\mathrm{ch}(\pi^{\D^{\R}}_{\sfF\oplus\NSR})=z^{\D^{\R}+1/16\;}\,\frac{\prod\limits_{k=0}^{\infty}(1+z^k)^2}{\prod\limits_{k=1}^{\infty}(1-z^k)}=
2\!\!\!\!\!\sum_{2n\in\mathbb{Z}+1/2}\!\!\!z^{\D^{\R}+2n^2-1/16}\prod\limits_{k=1}^{\infty}\frac{1}{(1-z^k)^2}=
\sum_{\substack{2n\in\mathbb{Z}+1/2\\
\epsilon=0,1}}\mathrm{ch}(\pi^{\epsilon,n}_{\Vir\oplus\Vir}).
\end{equation*}
where we used Jacobi triple product identity \eqref{JTP} in case $z\mapsto \sqrt{z}, y^2\mapsto \sqrt{z}$.
\end{proof}

\begin{Remark}\label{rem:proj}
We have isomorphism of modules $\pi^{\D^{\R},0}_{\sfF\oplus\NSR}\cong \pi^{\D^{\R},1}_{\sfF\oplus\NSR}$ so it is convenient to consider only one sector in calculations below. 
\end{Remark}

\subsection{Whittaker vector decomposition for $NS$ sector. Toda-like equations}
\label{ssec:Todaobt}
In \cite[Remark 4.1]{KGP} there were considered certain bilinear relations on conformal blocks.
In case of $c=1$ conformal blocks these relations do not provide differential equation on $\tau$ function given by \eqref{Kiev}.
But these relations provide differential-difference equation \eqref{difToda} on this $\tau$ function.
We first recall these bilinear relations, this part is a brief version of arguments in \cite[Sec. 4.2.]{KGP}. Then we 
deduce \eqref{difToda} from bilinear relations on conformal blocks, this part is not written in \cite{KGP} but is similar to arguments in loc. cit.
%But below we will see that these relations are sufficient for fact that \eqref{Kiev} is a solution of differential-difference equation
%\eqref{difToda} which was already proved in Subsection \ref{ssec:BTParam}.

The $\sfF \oplus \NSR$ Whittaker vector $|1\otimes W_{\NS}(z)\ra$ of $\BNS$ sector is defined as a tensor product of the $\sfF$ vacuum  and the $\NSR$ Whittaker vector. 
%The decomposition of the $\sfF\oplus \NSR$ representation  \eqref{FxNSR} provides a decomposition of the corresponding Whittaker vector
%\begin{equation}
%|1\otimes W_{\NS}(z)\rangle=\sum_{2n\in\mathbb{Z}}|v(z)\ra_n,
%\end{equation}
%where $|v(z)\ra_n\in\pi^n_{\Vir\oplus \Vir}$. It turns out that $|v(z)\ra_n$ is the Whittaker vector for the algebra $\Vir\oplus\Vir$
The following proposition was proven in \cite[Sec. 3.2]{KGP}.
\begin{prop}\label{prop:Whitdecomp}
The decomposition of the $\sfF\oplus \NSR$ Whittaker vector of $\BNS$ sector in terms of the subalgebra $\Vir\oplus\Vir$ has the form
\begin{equation} \label{Whitdecomp}
|1\otimes W_{\NS}(z)\ra=\sum\limits_{2n\in\mathbb{Z}}l_{n}(P,b)\,\Bigl(|W^{\scriptscriptstyle{(1)}}(\beta^{\scriptscriptstyle{(1)}}z)\ra_{n}\otimes|W^{\scriptscriptstyle{(2)}}(\beta^{\scriptscriptstyle{(2)}}z)\ra_{n}\Bigr).
\end{equation}
Here $|W^{\scriptscriptstyle{(1)}}\ra_{n}\otimes|W^{\scriptscriptstyle{(2)}}\ra_{n}$ denotes the tensor product of Whittaker vectors in $\pi^n_{\Vir\oplus\Vir}$, and the coefficients $l_n(P,b)$ do not depend on $z$. The parameters $\beta^{(\eta)}, \eta=1,2$ are defined by the formulae
\begin{equation}\label{eq:beta12}
\beta^{\scriptscriptstyle{(1)}}=\left(\frac{b^{-1}}{b^{-1}-b}\right)^2
\quad \beta^{\scriptscriptstyle{(2)}}=\left(\frac{b}{b-b^{-1}}\right)^2.
\end{equation}
\end{prop}
%\begin{proof}
%Let us act by $L^{(\eta)}_{1}, L^{(\eta)}_{2},  \eta=1,2$ on $|1\otimes W_{\NS}(z)\ra$. Using expressions \eqref{Vir12} and linear independence of vectors from different $\pi^n_{\Vir\oplus\Vir}$ we have
%\begin{equation*}
%\begin{aligned}
%L^{\scriptscriptstyle{(1)}}_{1}|v(z)\ra_n=(\beta^{\scriptscriptstyle{(1)}} z)|v(z)\ra_n, \quad L^{\scriptscriptstyle{(2)}}_{1}|v(z)\ra_n=(\beta^{\scriptscriptstyle{(2)}} z)|v(z)\ra_n, \\
%L^{\scriptscriptstyle{(1)}}_{2}|v(z)\ra_n=0, \quad
%L^{\scriptscriptstyle{(2)}}_{2}|v(z)\ra_n=0.
%\end{aligned}
%\end{equation*}
%Therefore the vector $|v(z)\ra_n$ is proportional to the tensor products of Whittaker vectors. Hence we proved ~\eqref{Whitdecomp}.
%\end{proof}

Introduce functions
\begin{equation}\label{eq:s:even:odd}
s_{\textrm{even}}(x,n)=
  \prod_{\substack{i,j\geq 0,\;i+j<2n\\ i+j\equiv0\bmod 2}}\hspace*{-10pt}(x+ib+jb^{-1}),\qquad
  s_{\textrm{odd}}(x,n)=2^{1/8}\hspace*{-10pt}
  \prod_{\substack{i,j\geq0,\;i+j<2n\\i+j\equiv1\bmod 2}}\hspace*{-10pt}(x+ib+jb^{-1}),
\end{equation}
for $n\geq 0$ and 
\begin{equation}
    s_{\textrm{even}}(x,n)=(-1)^n s_{\textrm{even}}(Q-x,-n),\quad
   s_{\textrm{odd}}(x,n)=s_{\textrm{odd}}(Q-x,-n)
\end{equation}
for $n<0$. Function $s_{\textrm{odd}}$ we will use in the Subsection \ref{ssec:Okamotoobt}.
The coefficients $l_n(P,b)$ in the formula \eqref{Whitdecomp} were calculated in \cite{HJ} and \cite[Sec. 3.3]{KGP}.
\begin{prop}
Coefficients $l_n(P,b)$ are given by 
\begin{equation}
l_n^2(P,b)=\frac{(-1)^{2n}2^{4n^2}(\beta^{\scriptscriptstyle{(1)}})^{-\D^{\scriptscriptstyle{(1)}}_n}(\beta^{\scriptscriptstyle{(2)}})^{-\D^{\scriptscriptstyle{(2)}}_n}}{s_{\textrm{even}}(2P,2n) s_{\textrm{even}}(2P+Q,2n)} \label{l_n_irr}
\end{equation}
\end{prop}
Note that normalization $\langle P,n|P,n\rangle=1$ determine vectors $|P,n\rangle$ only up to sign, therefore coefficients  $l_n(P,b)$ are also determined up to sign, But  $l_n(P,b)$ will appear in bilinear relation only as $l_n^2(P,b)$. 

Let us introduce operator $H$
\begin{equation}
H=b L_{0}^{\scriptscriptstyle{(1)}}+b^{-1} L_{0}^{\scriptscriptstyle{(2)}}\label{eq:Hdefin} 
\end{equation}
and define  $\widehat{\kr{F}_k}$ by the formulae
\begin{equation}\label{eq:Fhat}
\widehat{\kr{F}_k}=\la 1\otimes W_{\NS}(1)|H^{k}|1\otimes W_{\NS}(z)\ra
\end{equation}
We could calculate $\widehat{\kr{F}_k}$ in two ways.
First calculation is in terms of $\Vir\oplus\Vir$ generators and modules using Proposition \ref{prop:Whitdecomp}.
It gives
\begin{equation}
\label{eq:FhatHirota}
\widehat{\kr{F}_k}=\sum\limits_{2n\in\mathbb{Z}}l_n^2(P,b) \cdot D^{k}_{b,b^{-1}[\log z]}(\kr{F}^{\scriptscriptstyle{(1)}}_{n},\kr{F}^{\scriptscriptstyle{(2)}}_{n}),
\end{equation}
where we used shorter notation
$\kr{F}^{\scriptscriptstyle{(\eta)}}_{n}$ for
$\mathcal{F}_{c^{\scriptscriptstyle{(\eta)}}}(\Delta_n^{\scriptscriptstyle{(\eta)}}|\beta^{\scriptscriptstyle{(\eta)}}z)$ and generalized Hirota differential operators $D^{n}_{\epsilon_1,\epsilon_2[x]}$, $\epsilon_1, \epsilon_2\in\mathbb{C}$ are defined by
\begin{equation}
f(e^{\epsilon_1\alpha}z)g(e^{\epsilon_2\alpha}z)=\sum\limits_{n=0}^{\infty}D^{n}_{\epsilon_1,\epsilon_2[\log z]}(f(z),g(z))\frac{\alpha^n}{n!},\label{genhir}
\end{equation}
where we take derivatives with respect to the logarithm of variable as before. 

In other way using formulas \eqref{Vir12} we have
\begin{equation}
\label{eq:HNSR}
H=Q\sum\limits_{r\in\mathbb{Z}+1/2}r:f_{-r}f_{r}:-\sum\limits_{r\in\mathbb{Z}+1/2}f_{-r}G_r,
\end{equation}
Then we calculate $\widehat{\mathcal{F}_k}$ in terms of $\sfF\oplus\NSR$ generators and modules. It gives for $k=0,2$
\begin{equation}
\widehat{\kr{F}_0}=\kr{F}_{\NS},\qquad \widehat{\kr{F}_2}=-z^{1/2} \kr{F}_{\NS},
\label{eq:hatF}
\end{equation}
where we used shorter notation $\mathcal{F}_{\NS}$ for$\mathcal{F}_{c_{\nsr}}(\D^{\NS}|z)$.
Eliminating $\mathcal{F}_{\NS}$ from this equation and using \eqref{eq:FhatHirota} we have
\begin{equation}
\sum\limits_{2n\in\mathbb{Z}}l_n^2(P,b) \cdot D^{2}_{b,b^{-1}[\log z]}(\kr{F}^{\scriptscriptstyle{(1)}}_{n},\kr{F}^{\scriptscriptstyle{(2)}}_{n})
=-z^{1/2} \sum\limits_{2n\in\mathbb{Z}}l_n^2(P,b) \kr{F}^{\scriptscriptstyle{(1)}}_{n}\kr{F}^{\scriptscriptstyle{(2)}}_{n},
\label{Todablock}
\end{equation}
This is just relation written in \cite[Remark 4.1]{KGP}.

Now we deduce Toda-like equation \eqref{difToda} on $\tau$ function \eqref{Kiev} from the equation \eqref{Todablock}.
Let us substitute this $\tau$ function decomposition
into \eqref{difToda} and collect terms
with the same powers of $s$. The vanishing condition of the $s^m$
coefficient has the form
\begin{multline*}
\sum_{n\in\mathbb{Z}} C(\sg+n+m)C(\sg-n) \left(1/2D^2_{[\log z]}(\mathcal{F}((\sg+n+m)^2|z),\mathcal{F}((\sg-n)^2|z)+\right.\\+\left.z^{1/2}\mathcal{F}((\sg+1/2+n+m)^2|z)\mathcal{F}((\sg-1/2-n)^2|z)\right)=0
\end{multline*}
Clearly this $s^m$ coefficient coincides with the $s^{m+2}$
coefficient after the shift $\sg\mapsto \sg+1$.  Therefore it is sufficient
to prove the vanishing of $s^0$ and $s^1$ coefficients.

To obtain these bilinear relations we set in \eqref{Todablock} $c_{\nsr}=1 \Leftrightarrow b=i,\,c^{(\eta)}=1, \eta=1,2$.
We also substitute $P=2i\sg$, $z\mapsto 4z$.
Splitting \eqref{Todablock} into 
relations with powers $z^{2\sg^2+N}, N\in\mathbb{Z}$ and $z^{2\sg^2+N}, N\in\mathbb{Z}+1/2$
and shifting $n$ and $\sg$ in appropriate way we obtain correspondingly
\begin{multline*}
\sum_{n\in\mathbb{Z}}l_n^2(\sg) D^2_{[\log z]}(\kr{F}((\sg+n)^2|z),\kr{F}((\sg-n)^2|z))=\\
=2z^{1/2}\sum_{n\in\mathbb{Z}}l_{n+1/2}^2(\sg) \kr{F}((\sg+1/2+n)^2|z)\kr{F}((\sg-1/2-n)^2|z),
\end{multline*}
\begin{multline*}
\sum_{n\in\mathbb{Z}}l_{n+1/2}^2(\sg+1/2) D^2_{[\log z]}(\kr{F}((\sg+n+1)^2|z),\kr{F}((\sg-n)^2|z))=\\=2z^{1/2}
\sum_{n\in\mathbb{Z}}l_{n+1}^2(\sg+1/2) \kr{F}((\sg+3/2+n)^2|z)\kr{F}((\sg-1/2-n)^2|z),
\end{multline*}
where we used shorter notation $l_n(\sigma)$ for $l_n(2i\sg,i)$. It remains to compare coefficients of these relations and $s^0$, $s^1$ relations.
Using \eqref{l_n_irr} and shift relation on $\G(\sg)$ we get
\begin{equation}
\frac{C(\sg+n)C(\sg-n)}{C(\sg)^2}=\frac{1}{\prod\limits_{k=1}^{2|n|-1}(k^2-4\sg^2)^{2(2|n|-k)}(4\sg^2)^{2|n|}}=4^{-\D^{\NS}}(-1)^{2n}l_n^2(\sg) ,
\end{equation}
where $2n\in\mathbb{Z}$.
Then we have for $n\in\mathbb{Z}$
\begin{align*}
\frac{l_n^2(\sg)}{l_0^2(\sg)}&=\frac{C(\sg+n)C(\sg-n)}{C(\sg)^2},&\quad \frac{l^2_{n+\frac12}(\sg)}{l_0^2(\sg)}&=-\frac{C(\sg+\frac12+n)C(\sg-n-\frac12)}{C(\sg)^2},&\\
\frac{l_{n+\frac12}^2(\sg+\frac12)}{l_0^2(\sg)}&=-\frac{C(\sg+n+1)C(\sg-n)}{C(\sg+\frac12)^2},& \quad
\frac{l_{n+1}^2(\sg+\frac12)}{l_0^2(\sg)}&=\frac{C(\sg+n+\frac32)C(\sg-n-\frac12)}{C(\sg+\frac12)^2},&
\end{align*}
which completes the proof.

%\begin{Remark}
%It is interesting to mention that vector $z^{1/4}f_{-1/2}|1\otimes W_{\NS}(z)\rangle$ have similar to \ref{Whitdecomp}
%decomposition but with other coefficients $\td{l}_n(P,b)$.
%In subsection \ref{ssec:nWW} with the help of this vector we obtain new bilinear relations on conformal blocks.
%\end{Remark} 

\begin{Remark}
From results of section \ref{sec:Painleve} and Proposition \ref{taubackl} follows that the left side of \eqref{Kiev} satisfies \eqref{difToda} and this determines the function. Here we proved that the right side of \eqref{Kiev} satisfy \eqref{difToda} i.e. as a byproduct we proved \eqref{Kiev}.  This simplifies the proof of Theorem \ref{thm:1}
in comparison with \cite{KGP} because we do not consider equation \eqref{tau3} of order $4$.

On the other hand in Section \ref{sec:Painleve} we proved that function $\tau(z)$ and $\pi(\tau(z))$  satisfy Toda-like equations \eqref{eq:Toda}. Here we proved that $\tau(\sigma,s|z)$ and $\tau(\sigma-1/2,s|z)$ satisfy these Toda-like equations. So, as a byproduct we have another  proof of Proposition \ref{taubackl}. Note that the obtained proof is rather difficult since uses Theorem~\ref{thm:1}. 
\end{Remark}

\subsection{Whittaker vector decomposition for $R$ sector. Okamoto-like equations}
\label{ssec:Okamotoobt}
In this Subsection we will analogously to the $\BNS$ sector obtain bilinear relations on conformal blocks in the $\BR$ sector.
In terms of $\tau$ functions these relations for $c=1$ will have the same form as  Okamoto-like equations \eqref{difd8eq_1}, \eqref{difd8eq_2}. But we do not give representation theoretic proof of these equation since we do not find explicit formulas for coefficients $l_n$, see Conjecture \ref{Conj:ln} below. 

We have four $\sfF\oplus\NSR$ Whittaker vectors $z^{1/16}|1^{\mu}\otimes W_{\R,\nu}\ra$, $\mu, \nu=\pm$ of $\BR$ sector, 
where $1/16$ is the conformal dimension of the fermionic vacuum. These Whittaker vectors start from highest weight vectors $\overline{|\mu;\D^{\R},\nu\ra}$ with certain
parity $\mu\nu$ (notation $\mu\nu$ means product of signs in natural sense). Moreover, we have
\begin{lemma}\label{lemma}
Whittaker vectors $|W_{\R,\pm}\ra$ have certain parity
\begin{equation*}
\mathfrak{P}|W_{\R,\pm}\ra=\pm |W_{\R,\pm}\ra.
\end{equation*}
\end{lemma}
\begin{proof}
Vector $\mathfrak{P}|W_{\R,\pm}\ra$ satisfies properties \eqref{WhitaxRW} and starts from $\pm z^{\D^{\R}}|\D^{\R},\pm\ra$. 
As was mentioned in Subsection \ref{ssec:FxNSR} these properties and normalization condition determine Whittaker vector.  
\end{proof}

Then, analogously to Proposition \ref{prop:Whitdecomp} we have 
\begin{prop}
The decomposition of the $\sfF\oplus \NSR$ Whittaker vector of $\BR$ sector in terms of the subalgebra $\Vir\oplus\Vir$ has the form
\begin{equation}
z^{1/16}|1^{\mu}\otimes W_{\R,\nu}(z)\ra=\sum\limits_{2n+1/2\in\mathbb{Z}}l_{n}^{\mu,\nu}(P,b)\,\Bigl(|W^{\scriptscriptstyle{(1)}}(\beta^{\scriptscriptstyle{(1)}}z)\ra_{n}\otimes|W^{\scriptscriptstyle{(2)}}(\beta^{\scriptscriptstyle{(2)}}z)\ra_{n}\Bigr), \label{WhitdecompR}
\end{equation} 
Here $|W^{\scriptscriptstyle{(1)}}\ra_{n}\otimes|W^{\scriptscriptstyle{(2)}}\ra_{n}$ denotes the tensor product of Whittaker vectors of
$\pi^{n,\mu\nu}_{\Vir\oplus\Vir}$, and the coefficients $l_n^{\mu,\nu}(P,b)$ do not depend on $z$. The parameters $\beta^{(\eta)}, 
\eta=1,2$ are defined by the formulas \eqref{eq:beta12}.\label{prop:WhitdecompR}
\end{prop}

As in $\BNS$ sector coefficients $l_n^{\mu,\nu}$ are determined up to the sign.
But ratio of $l_n^{\mu,\nu}$ with the same parity is well defined because decomposition \eqref{WhitdecompR} starts with the same highest weight vectors.
Below we will reach relations only for even subspaces (see Remark \ref{rem:proj}). Then
\begin{prop}\label{prop:signs}
We have relations between $l_n^{+,+}(P,b)$ and $l_n^{-,-}(P,b)$
\begin{equation}
l_n^{-,-}(P,b)=(-1)^{2n-1/2\,}{}l_n^{+,+}(P,b).
\label{lconn}
\end{equation}
\end{prop}
\begin{proof}
Coefficients $l_n^{\mu,\nu}(P,b)$ are given by
\begin{equation*}
l_n^{\mu,\nu}(P,b)= {}_{\mu\nu}\la P,n|1^{\mu}\otimes W_{\R,\nu}(1)\ra.
\end{equation*}

We have isomorphism between module $\pi^{\D^{\R}}_{\sfF\oplus\NSR}$ with highest weight vectors $\overline{|\mu;\D^{\R},\nu\ra}$ and
momentum $P$ and the same module with transposed highest weight vectors $\overline{|\mu;\D^{\R},\nu\ra}'=
i\nu\overline{|-\mu;\D^{\R},-\nu\ra}$ and the same momentum $P$, $\mu,\nu=\pm$ (we denote this module by $\pi'^{\D^{\R}}_{\sfF\oplus\NSR}$). This isomorphism
follows from the properties \eqref{hwwF}, \eqref{hwwR} and normalization condition.  
Consider then $\sfF\oplus\NSR$ Whittaker vector $z^{1/16}|1^{+}\otimes W_{\R,+}(z)\ra$ in module $\pi^{\D^{\R}}_{\sfF\oplus\NSR}$.
Due to properties \eqref{WhitaxRW} this vector is also Whittaker vector $z^{1/16}|1^{-}\otimes W_{\R,-}(z)\ra$ in $\pi'^{\D^{\R}}_{\sfF\oplus\NSR}$
and its normalization is standard.

We have also that $\Vir\oplus\Vir$ highest weight vectors $|P,\pm1/4\ra'_{\epsilon}$, $\epsilon=0,1$ in $\pi'^{\D^{\R}}_{\sfF\oplus\NSR}$
given by equality in \eqref{starthwwR} are connected with analogously defined $\Vir\oplus\Vir$ highest weight vectors 
in module $\pi^{\D^{\R}}_{\sfF\oplus\NSR}$ by
\begin{equation*}
|P,\pm1/4\ra'_{\epsilon}=\pm\epsilon|P,\pm 1/4\ra_{\epsilon}, \quad \epsilon=0,1.
\end{equation*}
where we used \eqref{starthwwR}.
Then using \eqref{hwwdecR} we have that
\begin{equation*}
|P,n+1/4\ra'_0=(-1)^{2n}|P,n+1/4\ra_0,
\end{equation*}
where $2n\in\mathbb{Z}$.
%and fact that $|P,n\ra_{\epsilon}=|-P,-n\ra_{\epsilon}$ which follows from construction of free-field representation \eqref{eq.FFR}.
Then
\begin{equation*}
l_{n}^{-,-}(P,b)={}'\la P,n|1^{-}\otimes W_{\R,-}(1)\ra=
(-1)^{2n-1/2\,}{}\la P,n|1^{+}\otimes W_{\R,+}(1)\ra=(-1)^{2n-1/2\,}{}l_n^{+,+}(P,b)
\end{equation*}
which completes the proof.
\end{proof}

Coefficients $l_n^{+,+}(P,b)$ possibly could be calculated in the way analogous to that in \cite[Sec.3.3.]{KGP} or in \cite{HJ}. It is natural to conjecture 
\begin{conj} \label{Conj:ln}
Coefficient $l_n^{+,+}(P,b)$ is given by expression
\begin{equation}
l_n^{+,+}(P,b)^2=\frac{2^{4n^2-1}(\beta^{\scriptscriptstyle{(1)}})^{-\D^{\scriptscriptstyle{(1)}}_n}(\beta^{\scriptscriptstyle{(2)}})^{-\D^{\scriptscriptstyle{(2)}}_n}}{s_{\textrm{odd}}(2P,2n) s_{\textrm{odd}}(2P+Q,2n)}, \label{l_n_irrR}
\end{equation} 
where function $s_{\textrm{odd}}$ was defined in \eqref{eq:s:even:odd}.
\end{conj}
This was checked by computer calculations for $|n|\leq 9/4$. More precisely, we check relations first relation in \eqref{Okamotoblock} with these coefficients up to $z$ in power $\Delta_{\R}+1/16+10$. 

%In this approach relations (4.49) remain nontrivial except cases of $m=2n^2-1/8$.

Now we want to calculate $\widehat{\mathcal{F}_k}$ defined by
\begin{equation}
\widehat{\mathcal{F}_k}^{\mu,\nu}=z^{1/16}\la 1^{\mu}\otimes W_{\R,\nu}(1)|H^{k}|1^+\otimes W_{\R,+}(z)\ra, \label{eq:FhatR}
\end{equation}
where even operator $H$ is given by \eqref{eq:Hdefin}, $\mu,\nu=\pm$.
We choose ket Whittaker vector with signs $+$ because only relative signs of bra and ket Whittakers are interesting (this follows from isomorphism mentioned in proof of Proposition \ref{prop:signs} and analogous
isomorphism obtained by transposition $\overline{|\mu;\D^{\R},\nu\ra}'=\overline{|\mu;\D^{\R},-\nu\ra}$).
It follows from Lemma \ref{lemma} that only functions $\widehat{\mathcal{F}_k}^{\mu,\mu}$ are nonzero and
we will denote them simply by $\widehat{\mathcal{F}_k}^{\mu}$.

We can calculate $\widehat{\mathcal{F}_k}$ using r.h.s. of \eqref{WhitdecompR}
\begin{multline}
\sum_{k=0}^\infty \widehat{\mathcal{F}}_k\frac{\alpha^k}{k!}=z^{1/16}\la 1^{\mu}\otimes W_{\R,\mu}(1)|e^{\alpha H}|1^+\otimes W_{\R,+}(z)\ra=\\=\!\!\!\!\!\!\!\sum_{2n+1/2\in\mathbb{Z}}\!\!\!\!\!\!\!l_n^{\mu,\mu}(P,b) l_n^{+,+}(P,b)
\la W^{(1)}_n(\beta^{(1)})|e^{\alpha bL_0^{(1)}}|W^{(1)}_n(\beta^{(1)}z)\ra \la W^{(2)}_n(\beta^{(2)})|e^{\alpha b^{-1}L_0^{(2)}}|W^{(2)}_n(\beta^{(2)}z)\ra=\\=
\sum_{2n+1/2\in\mathbb{Z}}l_n^{\mu,\mu}(P,b)l_n^{+,+}(P,b)\mathcal{F}_n^{(1)}(\beta^{(1)}ze^{\alpha b})\mathcal{F}_n^{(2)}(\beta^{(2)}ze^{\alpha b^{-1}})=\\
=\sum_{k=0}^{\infty}\sum_{2n+1/2\in\mathbb{Z}}l_n^{\mu,\mu}(P,b)l_n^{+,+}(P,b)D^k_{b,b^{-1}[\log z]}(\mathcal{F}_n^{(1)}(\beta^{(1)}z),\mathcal{F}_n^{(2)}(\beta^{(2)}z))\frac{\alpha^k}{k!},
\label{FhatR}
\end{multline}
where we used definition of generalized Hirota differential \eqref{genhir}.

%Choosing which submodule is even or odd in $\pi^{\R}_{\sfF\oplus\NSR}$ is unessential. That's because
%there exist an isomorphism $\mathfrak{O}_{\NSR}$ between both submodules which interchange highest weight vectors $|\D^{\R},+\ra\leftrightarrow|\D^{\R},-\ra$ 
%and commutes with $\sfF\oplus\NSR$ algebra. 

%Analogously to Proposition \ref{Whitdecomp} we have decomposition

%where $|W^{\scriptscriptstyle{(1)}}\ra_n\otimes |W^{\scriptscriptstyle{(2)}}\ra_n \in \pi^{\epsilon,n}_{\Vir\oplus\Vir}, \quad \mu,\nu=\pm, \quad \epsilon=\frac{1-\mu\nu}{2}$ 
%with the same $\beta^{(\eta)}, \eta=1,2$.  We will see below that $l_n^{\mu,\nu}(P,b), \mu,\nu=\pm$ are simply connected. 

%Coefficients $l_n^{\mu,\nu}(P,b)$ are defined up to the sign. We will partially fix this sign by $l_n(P,b)^{\mu,\pm}=l_n(P,b)^{\mu,\mp}$. We can do this due to
%isomorphism $\mathfrak{O}_{\NSR}$. There exists other isomorphism $\mathfrak{O}_{\sfF}=\sqrt{-2} f_0$, which interchanges $|1^+\ra\leftrightarrow |1^-\ra$ and $P\rightarrow -P$, 
%from which existence it follows that
%\begin{equation}
%l_n^{+,\nu}(P,b)=l_n^{-,\nu}(-P,b)
%\end{equation}

On the other hand we could rewrite $H$ in terms of $\sfF\oplus\NSR$ generators using \eqref{Vir12}
\begin{equation}
\label{eq:HNSRR}
H=Q\left(\sum\limits_{r\in\mathbb{Z}}r:f_{-r}f_{r}:+1/8\right)-\sum\limits_{r\in\mathbb{Z}}f_{-r}G_r,
\end{equation}

%Due to above mentioned and commutativity of operator $H=bL_0^{(1)}+b^{-1}L_0^{(2)}$ with $\mathfrak{P}$ we will use below only even Whittaker vectors.
%It appears that this is only nontrivial case, because $\la W_{\R,\mp} |W_{\R,\pm} \ra \propto \la \D^{\R} \mp|\D^{\R}\pm\ra=0$.

%Let us obtain equations in the same way, as in $\BNS$ case. In one hand analogs of \eqref{eq:FhatHirota} hold.
%Operator $H$ in terms of the $\mathsf{F}\oplus\NSR$ generators

We will use calculations of $(H-Q/8)^k z^{1/16}|1^{\pm}\otimes W_{\R,\pm}(z)\ra$ for $k\leq 3$
\begin{align*}
(H-Q/8)z^{1/16}|1^{+}\ra\otimes |W_{\R,+}\ra=&-f_0 G_0 z^{1/16}|1^{+}\ra\otimes|W_{\R,+}\ra\\
(H-Q/8)^2z^{1/16}|1^{+}\ra\otimes |W_{\R,+}\ra=&-1/2(L_0-c_{\nsr}/16)z^{1/16}|1^{+}\ra\otimes |W_{\R,+}\ra
-\frac{z}{\sqrt2}f_{-1}z^{1/16}|1^{-}\ra\otimes |W_{\R,+}\ra\\
(H-Q/8)^3z^{1/16}|1^{+}\ra\otimes |W_{\R,+}\ra=&1/2f_0G_0 (L_0-c_{\nsr}/16)z^{1/16}|1^{+}\ra\otimes |W_{\R,+}\ra-
-Q\frac{z}{\sqrt2}f_{-1}z^{1/16}|1^{-}\ra\otimes |W_{\R,+}\ra-\\-\frac{z}{\sqrt2}G_{-1}z^{1/16}|1^{-}\ra\otimes |W_{\R,+}\ra+&
\frac{z}{2}f_{-1}z^{1/16}|1^{+}\ra\otimes G_0 |W_{\R,+}\ra
%\\
%(H-Q/8)^4z^{1/16}|1^{+}\ra\otimes |W_{\R,+}\ra=&1/4 (L_0-c_{\nsr}/16)^2 z^{1/16}|1^{+}\ra\otimes |W_{\R,+}\ra+
%\frac{z}{\sqrt2}f_0 G_0G_{-1}z^{1/16}|1^{-}\ra\otimes |W_{\R,+}\ra+
% \ldots,
\end{align*}
%where \ldots stays for terms, which vanish after multiplication on $\la 1^{+}|\otimes |W_{\R,+}|$. 
And after the multiplication on $\la 1^{+}|\otimes |W_{\R,+}|$ we get 
\begin{equation}
\begin{aligned}
\widehat{\mathcal{F}_0}^{'+}&=z^{1/16}\mathcal{F}_{\R},&\qquad \widehat{\mathcal{F}_1}^{'-}&=-\frac{iP}2 z^{1/16} \kr{F}_{\R},\\
\widehat{\mathcal{F}_2}^{'+}&=-1/2z^{1/16}\left(z\frac{d}{dz}-c_{\nsr}/16\right)\kr{F}_{\R},& \qquad \widehat{\mathcal{F}_3}^{'-}&=
-\frac{iP}{4}z^{1/16}\left(z\frac{d}{dz}-c_{\nsr}/16\right)\kr{F}_{\R}, 
\label{eq:hatFR}
\end{aligned}
\end{equation}
where we denote by $\widehat{\mathcal{F}_k}^{'\mu}$ modified $\widehat{\mathcal{F}_k}^{\mu}$ with $H\mapsto H-Q/8$ in definition \eqref{eq:FhatR}.
We also shorten notation $\kr{F}_{c_{\nsr}}(\Delta^{\R}|z)$ to $\kr{F}_{\R}$.

Now we are interested in two relations following from \eqref{eq:hatFR} by elimination of $\kr{F}_{\R}$
\begin{equation}
\begin{aligned}\label{Okamotoblock}
\widehat{\mathcal{F}_2}^{'+}=-\frac12\left(z\frac{d}{dz}-c_{\nsr}/16-1/16\right)\widehat{\mathcal{F}_0}^{'+} \\
\widehat{\mathcal{F}_3}^{'-}=-\frac12\left(z\frac{d}{dz}-c_{\nsr}/16-1/16\right)\widehat{\mathcal{F}_1}^{'-}.
	\end{aligned}
\end{equation}
Due to \eqref{FhatR} these relations can be viewed as a bilinear differential relations on $\Vir$ conformal blocks.

As in $\BNS$ sector we set $c_{\nsr}=1 \Leftrightarrow b=i,\,c^{(\eta)}=1,
\eta=1,2$. We also substitute $P=2i\sg$, $z\mapsto 4z$. In this specialization $\widehat{\mathcal{F}_k}^{'\mu}=\widehat{\mathcal{F}_k}^{\mu}$ and using \eqref{FhatR} the relations \eqref{Okamotoblock} 
\eqref{FhatR} turn to
\begin{multline*}
\sum_{2n+1/2\in\mathbb{Z}}l_n^{+,+}(\sg)^2D^2_{[\log z]}(\mathcal{F}((\sg+n)^2|z),\mathcal{F}((\sg-n)^2|z))=\\=-\frac12\left(z\frac{d}{dz}-\frac18\right)\sum_{2n+1/2\in\mathbb{Z}}l_n^{+,+}(\sg)^2\mathcal{F}((\sg+n)^2|z)\mathcal{F}((\sg-n)^2|z),
\end{multline*} \vspace{-0.6cm}
\begin{multline*}
\sum_{2n+1/2\in\mathbb{Z}}(-1)^{2n-1/2}l_n^{+,+}(\sg)^2D^3_{[\log z]}(\mathcal{F}((\sg+n)^2|z),\mathcal{F}((\sg-n)^2|z))=\\
=
-\frac12\left(z\frac{d}{dz}-\frac18\right)\sum_{2n+1/2\in\mathbb{Z}}(-1)^{2n-1/2}l_n^{+,+}(\sg)^2D^1_{[\log z]}(\mathcal{F}((\sg+n)^2|z),\mathcal{F}((\sg-n)^2|z)).
\end{multline*}
where we used \eqref{lconn} and shorten notation $l_n^{+,+}(2i\sg,i)$ to $l_n^{+,+}(\sg)$.
Since $l^2_{-n}(\sg)=l^2_n(\sg)$ (due to \eqref{l_n_irrR} in case $Q=0$) we can divide these sums on 2 and get 
\begin{multline}
\sum_{n\in\mathbb{Z}}l_{n+1/4}^{+,+2}(\sg)D^2_{[\log z]}(\mathcal{F}((\sg+n+1/4)^2|z),\mathcal{F}((\sg-n-1/4)^2|z))=\\=
\frac12\left(z\frac{d}{dz}-\frac18\right)\sum_{n\in\mathbb{Z}}l_{n+1/4}^{+,+}(\sg)^2\mathcal{F}((\sg+n+1/4)^2|z)\mathcal{F}((\sg-n-1/4)^2|z),\label{d8rel1}
\end{multline}\vspace{-0.6cm}
\begin{multline}
\sum_{n\in\mathbb{Z}}l_{n+1/4}^{+,+2}(\sg)D^3_{[\log z]}(\mathcal{F}((\sg+n+1/4)^2|z),\mathcal{F}((\sg-n-1/4)^2|z))=\\=
-\frac12\left(z\frac{d}{dz}-\frac18\right)\sum_{n\in\mathbb{Z}}l_{n+1/4}^{+,+2}(\sg)D^1_{[\log z]}(\mathcal{F}((\sg+n+1/4)^2|z),\mathcal{F}((\sg-n-1/4)^2|z)),\label{d8rel2}
\end{multline}

Now we deduce Okamoto-like equations \eqref{difd8eq_1}, \eqref{difd8eq_2} on $\tau$ function \eqref{Kiev}.
Let us substitute this $\tau$ function decomposition
into \eqref{difd8eq_1} and \eqref{difd8eq_2} and collect terms
with the same powers of $s$. The vanishing condition of the $s^m$
coefficient has the form
\begin{multline*}
\sum_{n\in\mathbb{Z}} C(\sg+n+m)C(\sg-1/2-n) D^2_{[\log z]}(\mathcal{F}((\sg+n+m)^2|z),\mathcal{F}((\sg-1/2-n)^2|z))=\\=
\frac12\left(z\frac{d}{dz}-\frac18\right)\sum_{n\in\mathbb{Z}} C(\sg+n+m)C(\sg-1/2-n) \mathcal{F}((\sg+n+m)^2|z)\mathcal{F}((\sg-1/2-n)^2|z),
\end{multline*}\vspace{-0.6cm}
\begin{multline*}
\sum_{n\in\mathbb{Z}} C(\sg+n+m)C(\sg-1/2-n) D^3_{[\log z]}(\mathcal{F}((\sg+n+m)^2|z),\mathcal{F}((\sg-1/2-n)^2|z))=\\
=
\frac12\left(z\frac{d}{dz}-\frac18\right)\sum_{n\in\mathbb{Z}} C(\sg+n+m)C(\sg-1/2-n) D^1_{[\log z]}(\mathcal{F}((\sg+n+m)^2|z)\mathcal{F}((\sg-1/2-n)^2|z)).
\end{multline*}
Clearly this $s^m$ coefficient coincides with the $s^{m+1}$
coefficient after the shift $\sg\mapsto \sg+1/2$.  Therefore it is sufficient
to prove the vanishing of $s^0$ coefficient. 

Let us use the substitution $\sg\mapsto \sg+1/4$ in the $s^0$ relation 
\begin{equation*}
\begin{aligned}
&\sum_{n\in\mathbb{Z}} C(\sg+n+1/4)C(\sg-1/4-n) D^2_{[\log z]}(\mathcal{F}((\sg+n+1/4)^2|z),\mathcal{F}((\sg-n-1/4)^2|z))=\\&=
\frac12\left(z\frac{d}{dz}-1/8\right)\sum_{n\in\mathbb{Z}} C(\sg+n+1/4)C(\sg-1/2-n) \mathcal{F}((\sg+n+1/4)^2|z)\mathcal{F}((\sg-n-1/4)^2|z),
\end{aligned}
\end{equation*}
\begin{equation*}
\begin{aligned}
&\sum_{n\in\mathbb{Z}} C(\sg+n+1/4)C(\sg-n-1/4) D^3_{[\log z]}(\mathcal{F}((\sg+n+1/4)^2|z),\mathcal{F}((\sg-n-1/4)^2|z))=\\&=
\frac12\left(z\frac{d}{dz}-1/8\right)\sum_{n\in\mathbb{Z}} C(\sg+n+1/4)C(\sg-n-1/4) D^1_{[\log z]}(\mathcal{F}((\sg+n+1/4)^2|z)\mathcal{F}((\sg-n-1/4)^2|z)).
\end{aligned}
\end{equation*}
These relations coincide with relations \eqref{d8rel1}, \eqref{d8rel2} up to coefficients.
Using shift relation on $\G(\sg)$ and \eqref{l_n_irrR} we obtain 
\begin{equation}
\frac{C(\sg+n+1/4)C(\sg-n-1/4)}{C(\sg+1/4)C(\sg-1/4)}=\frac{1}{\prod\limits_{k=0}^{2|n+1/4|-3/2}((1/2+k)^2-4\sg^2)^{2(2|n+1/4|-1/2-k)}}=
2^{7/8}4^{-\D^R}l^{+,+2}_{n+1/4}(\sg),
\end{equation}
which completes the proof.

%\section{Further questions}
\section{Acknowledgments}
We thank  P. Gavrylenko, A. Its, O. Lisovyy, A. Marshakov, H. Sakai and A. Sciarappa  for interest in
our work and discussions.

This work has been funded by the Russian Academic Excellence Project
’5-100’.  A.S. was also supported in part by joint NASU-CNRS project F14-2016, M.B. was also supported in part  by Young Russian Mathematics award and RFBR grant mol\_a\_ved 15-32-20974.
Study of algebraic solutions was performed under the grant of Russian Science Foundation 14-050-00150 for IITP.

\noindent \textsc{Landau Institute for Theoretical Physics, Chernogolovka, Russia,\\
	Skolkovo Institute of Science and Technology, Moscow, Russia,\\
	National Research University Higher School of Economics, Moscow, Russia,\\
	Institute for Information Transmission Problems, Moscow, Russia,\\
	Independent University of Moscow, Moscow, Russia}

\emph{E-mail}:\,\,\textbf{mbersht@gmail.com}\\

\noindent\textsc{National Research University Higher School of Economics, Moscow, Russia\\
	Skolkovo Institute of Science and Technology, Moscow, Russia,\\
	Bogolyubov Institute for Theoretical Physics, Kiev, Ukraine}

\emph{E-mail}:\,\,\textbf{shch145@gmail.com}
\end{document}